\documentclass[fleqn,thmsb,a4paper,12pt]{article}
\usepackage{amssymb}
\usepackage{amsmath}
\usepackage{color}
\usepackage{amsthm}
\usepackage{amssymb}
\usepackage{amsfonts}
\usepackage{amstext}
\usepackage[USenglish]{babel}

\newtheorem{acorollary}{Corollary}

\newtheorem{atheorem}{Theorem}
\theoremstyle{definition}
\newtheorem{adefinition}{Definition}
\newtheorem{aexample}{Example}
\newtheorem{aremark}{Remark}
\begin{document}

\author{Ata Atay\thanks{Department of Mathematical Economics, Finance and Actuarial Sciences, University of Barcelona, Spain. E-mail: aatay@ub.edu} \and Ana Mauleon\thanks{CEREC and CORE/LIDAM, UCLouvain, Belgium. E-mail: ana.mauleon@usaintlouis.be} \and Vincent Vannetelbosch\thanks{CORE/LIDAM, UCLouvain, Belgium. E-mail: vincent.vannetelbosch@uclouvain.be}}
\title{\textbf{Limited Farsightedness in \\
Priority-Based Matching}}
\date{December 12, 2022}
\maketitle

\begin{abstract}
We consider priority-based matching problems with limited farsightedness. We show that, once agents are sufficiently farsighted, the matching obtained from the Top Trading Cycles (TTC) algorithm becomes stable: a singleton set consisting of the TTC matching is a horizon-$k$ vNM stable set if the degree of farsightedness is greater than three times the number of agents in the largest cycle of the TTC. On the contrary, the matching obtained from the Deferred Acceptance (DA) algorithm may not belong to any horizon-$k$ vNM stable set for $k$ large enough.\\

\end{abstract}

\thispagestyle{empty}\newpage

\pagenumbering{arabic}

\section{Introduction}\label{S1}

Many objects such as houses, school seats, jobs or organs are allocated based on the preferences of the agents and their priorities.\footnote{Roth and Sotomayor (1990) and Haeringer (2017) provide a general introduction to matching problems.} Two approaches for analysing the stability of matchings have been proposed in the literature depending on whether and how far agents anticipate that their actions may also induce others to change their matches. On the one hand, standard stability concepts involves fully myopic agents in the sense that they do not anticipate that others might react to their actions. On the other hand, a number of solution concepts involve perfectly farsighted agents who fully anticipate the complete sequence of reactions that results from their own actions. However, experimental evidence suggests that subjects are consistent with an intermediate degree of farsightedness: agents only anticipate a limited number of reactions by the other agents to the actions they take themselves.\footnote{See e.g. Kirchsteiger, Mantovani, Mauleon and Vannetelbosch (2016).} 

Two  prominent mechanisms used for priority-based matching are the Gale and Shapley's (1962) Deferred Acceptance (DA) mechanism and the Shapley and Scarf's (1974) Top Trading Cycles (TTC) mechanism. Abdulkadiro\u{g}lu and Sönmez (2003) show that both mechanisms are strategy-proof: truthful preference revelation is a weakly dominant strategy for the agents.\footnote{Dubins and Freedman (1981) and Roth (1982) were first to show that the DA mechanism satisfies strategy-proofness in one-to-one matching problems.}  But, on the one hand, the TTC mechanism is Pareto efficient while the DA mechanism may select an inefficient matching. On the other hand, the DA mechanism is stable while the TTC mechanism may select an unstable matching.\footnote{Reny (2022) introduces the Priority-Efficient (PE) mechanism that always selects a Pareto efficient matching that dominates the DA stable matching, but PE is not strategy-proof. See also Kesten (2006, 2010).} 

In one-to-one priority-based matching problems, is it possible to stabilize the matching obtained from the TTC algorithm when agents are limited farsighted? If yes, how much farsightedness from the agents do we need? To answer these questions we propose the notion of horizon-$k$ vNM stable set to study the matchings that are stable when agents are limited in their degree of farsightedness.\footnote{See Chwe (1994), Mauleon, Vannetelbosch and Vergote (2011), Ray and Vohra (2015, 2019), Herings, Mauleon and Vannetelbosch (2019, 2020), Luo, Mauleon and Vannetelbosch (2021) for definitions of the farsighted stable set.} A horizon-$k$ improving path for priority-based matching problems is a sequence of matchings that can emerge when limited farsighted agents form or destroy matches based on the improvement the $k$-steps ahead matching offers them relative to the current matching. A set of matchings is a horizon-$k$ vNM stable set if (Internal Stability) for any two matchings belonging to the set, there is no horizon-$k$ improving path from one matching to the other one, and (External Stability) there always exists a horizon-$k$ improving path from every matching outside the set to some matching within the set.

We show that, once agents are sufficiently farsighted, the matching obtained from the TTC algorithm becomes stable in one-to-one priority-based matching problems. Precisely, a singleton set consisting of the TTC matching is a horizon-$k$ vNM stable set if the degree of farsightedness is greater than three times the number of agents in the largest cycle of the TTC. We provide a constructive proof where we build step by step a horizon-$k$ improving path from any matching leading to the TTC matching. Along the horizon-$k$ improving path, agents move one at a time and agents belonging to cycles sequentially act in the order of the formation of cycles in the TTC algorithm. Looking forward $k$ steps ahead, agents belonging to a cycle first match one by one to the object that ranks them first on their priority list. Second, they give up one by one their object, and by doing so, vacating the object. Third, they match one by one to the object they are assigned to in the TTC matching. The number of steps in this improving path is at most equal to three times the number of agents in the largest cycle of the TTC. Hence, looking forward such a number of steps ahead allows the agents to anticipate ending up with their TTC matches; and by doing so, they have incentives for engaging a move towards the matches they have in the TTC matching. Finally, we show that, in the special case where each agent owns an object, a singleton set consisting of the TTC matching is the unique horizon-$k$ vNM stable set. 

Thus, the matching obtained from the TTC algorithm is not only Pareto efficient and strategy-proof, it is also horizon-$k$ stable. On the contrary, the matching obtained from the DA algorithm may not belong to any horizon-$k$ vNM stable set for $k$ large enough. Next, we provide a tighter bound on the length of the largest cycle that is sufficient for stabilizing the TTC matching. However, this tighter bound requires more coordination on behalf of the agents. We also show that our main result is robust to an alternative concept for limited farsightedness, which is obtained by adapting Herings, Mauleon and Vannetelbosch (2019) definition of a horizon-$L$ farsighted set of networks to priority-based matching problems. This concept mainly replaces the internal stability condition of the horizon-$k$ vNM stable set by two alternative conditions: deterrence of external deviations and minimality.

Our results complement Atay, Mauleon and Vannetelbosch (2022) who show that for school choice problems with farsighted students, a singleton set consisting of the TTC matching is a farsighted stable set.\footnote{Morrill (2015) and Hakimov and Kesten (2018) propose variations of the TTC for school choice problems. Atay, Mauleon and Vannetelbosch (2022) show that the matchings obtained from the variations are all farsightedly stable. For one-to-one priority-based problems, all variations coincide.}  In addition, they also find that the DA matching may not belong to any farsighted stable set, and so the TTC mechanism provides an assignment that is not only Pareto efficient but also farsightedly stable for many-to-one priority-based matching problems. Notice that Atay, Mauleon and Vannetelbosch (2022) allow group of students (and schools) to move all together. Here, we show that, in one-to-one priority-based matching problems, their main result is robust to one agent moving at a time and limited farsightedness. 

There is a recent literature that provides theoretical and/or empirical results supporting the TTC mechanism in one-to-one priority-based matching problems. For instance, Abdulkadiro\u{g}lu, Che, Pathak, Roth and Tercieux (2020) show that the TTC mechanism is justified envy minimal in the class of Pareto efficient and strategy-proof mechanisms. In addition, Do\u{g}an and Ehlers (2022) find that, for any stability comparison satisfying three basic properties, the TTC mechanism is minimally unstable among Pareto efficient and strategy-proof mechanisms.\footnote{Do\u{g}an and Ehlers (2021) study efficient and minimally unstable Pareto improvements over the DA mechansism.} 

The paper is organized as follows. In Section \ref{S2}, we introduce priority-based matching problems. In Section \ref{S3}, we provide a formal description of the TTC mechanism and its algorithm. In Section \ref{S4}, we introduce the notions of horizon-$k$ improving path and horizon-$k$ stable set for priority-based matching problems. In Section \ref{S5} we provide our main result. In Section \ref{S6}, we discuss some extensions: (i) a tighter bound on the degree of farsightedness, (ii) all agents are farsighted, (iii) an alternative concept of limited farsightedness, and (iv) the matching problem where each agent owns an object.

\section{Priority-based matching problems}\label{S2} 

A priority-based matching problem is a list $\langle I,S,P,F\rangle$ where 
\begin{itemize}
\item[(\textbf{i})] $I=\{i_1,...,i_n\}$ is the set of agents,
\item[(\textbf{ii})] $S=\{s_1,...,s_m\}$ is the set of objects,
\item[(\textbf{iv})] $P=(P_{i_1},...,P_{i_n})$ is the preference profile where $P_i$ is the strict preference of agent $i$ over the objects and her outside option,
\item[(\textbf{v})] $F=(F_{s_1},...,F_{s_m})$ is the strict priority structure of the objects over the agents.
\end{itemize}

Let $i$ be a generic agent and $s$ be a generic object. The preference $P_i$ of agent $i$ is a linear order over $S \cup \{i\}$. Agent $i$ prefers object $s$ to object $s^\prime$ if $s P_i s^\prime $. Object $s$ is acceptable to agent $i$ if $s P_i i$ (i.e. $s P_i i $ means that $i$ strictly prefers $s$ to being unassigned). We often write $P_i = s,s^\prime,s^{\prime \prime}$ meaning that agent $i$'s most preferred object is $s$, her second best is $s^\prime$, her third best is $s^{\prime \prime}$ and any other object is unacceptable for her. Let $R_i$ be the weak preference relation associated with the strict preference relation $P_i$.

The priority $F_s$ of object $s$ is a linear order over $I$. That is, $F_s$ assigns ranks to agents according to their priority for object $s$. The rank of agent $i$ for object $s$ is denoted by $F_s(i)$ and $F_s(i) < F_s(j)$ means that agent $i$ has higher priority for object $s$ than agent $j$. For $s \in S, i \in I$, let $\Phi(s,i) = \{j \in I \mid F_s(j) < F_s(i)\}$ be the set of agents who have higher priority than agent $i$ for object $s$.

An outcome of a priority-based matching problem is a matching $\mu : I \rightarrow S \cup I$ such that for any $i \in I$ and any $s \in S$,
\begin{itemize}
\item[(\textbf{i})]$\mu(i) \notin S \Rightarrow \mu(i)=i$, and
\item[(\textbf{ii})]$\#(\mu^{-1}(s)) \leq 1$.
\end{itemize}

Condition (i) means that agent $i$ is assigned to object $s$ under $\mu$ if $\mu(i)=s$ and is unassigned under $\mu$ if $\mu(i)=i$. Condition (ii) requires that no object is assigned to more than one agent. The set of all matchings is denoted $\mathcal{M}$.\footnote{Throughout the paper we use the notation $\subseteq$ for weak inclusion and $\subset$ for strict inclusion. Finally, $\#$ will refer to the notion of cardinality.} For instance, 
\[\mu=\Big( \begin{array}{cccc} i_1 & i_2 & i_3 & i_4 \\ s_2 & s_3 & s_1 & i_4\end{array}\Big)\] is the matching where agent $i_1$ is assigned to object $s_2$, agent $i_2$ is assigned to object $s_3$, agent $i_3$ is assigned to object $s_1$, and agent $i_4$ is unassigned. For convenience, we often write such matching as $\mu = \{(i_1,s_2),(i_2,s_3),(i_3,s_1),(i_4,i_4)\}$.

A matching $\mu^\prime$ Pareto dominates a matching $\mu$ if $\mu^\prime (i) R_i \mu(i)$ for all $i \in I$ and $\mu^\prime (j) P_j \mu(j)$ for some $j \in I$. A matching is Pareto efficient if it is not Pareto dominated by another matching. Let $\mathcal{E}(I,S,P,F)$ be the set of Pareto efficient matchings. A matching $\mu$ is stable if 
\begin{itemize}
\item[(\textbf{i})]for all $i \in I$ we have $\mu(i) R_i i$ (individual rationality),
\item[(\textbf{ii})]for all $i \in I$ and all $s \in S$, if $ s P_i \mu(i)$ then $\#(\mu^{-1}(s)) = 1$ (non-wastefulness),
\item[(\textbf{iii})]for all $i,j \in I$ with $\mu(j)=s$, if $\mu(j) P_i \mu(i)$ then $j \in \Phi(s,i)$ (no justified envy).
\end{itemize}
Let $\mathcal{S}(I,S,P,F)$ be the set of stable matchings.

A mechanism systematically selects a matching for any given priority-based matching problem $\langle I,S,P,F \rangle$.  A mechanism is individually rational (non-wasteful / stable / Pareto efficient) if it always selects an individually rational (non-wasteful / stable / Pareto efficient) matching. A mechanism is strategy-proof if no agent can ever benefit by unilaterally misrepresenting her preferences.

\section{The Top Trading Cycles algorithm}\label{S3} 

Abdulkadiro\u{g}lu and Sönmez (2003) introduce the Top Trading Cycles (TTC) mechanism for selecting a matching for general priority-based matching problems. In the case of a priority-based matching problem $\langle I,S,P,F \rangle$, the TTC mechanism finds a matching by means of the following TTC algorithm.

\begin{itemize}
\item[Step $1$.] Each agent $i \in I$ points to the object that is ranked first in $P_i$. If there is no such object, then agent $i$ points to herself and she forms a self-cycle. Each object $s \in S$ points to the agent that has the highest priority in $F_s$. Since the number of agents and objects are finite, there is at least one cycle. A cycle is an ordered list of distinct objects and distinct agents $(s^1,i^1,s^2,...,s^l,i^l)$ where $s^1$ points to $i^1$ (denoted $s^1 \mapsto i^1$), $i^1$ points to $s^2$ ($i^1 \mapsto s^2$), $s^l$ points to $i^l$ ($s^l \mapsto i^l$) and $i^l$ points to $s^1$ ($i^l \mapsto s^1$). Each object (agent) can be part of at most one cycle. Every agent in a cycle is assigned to the object she points to and she is removed. Similarly, every agent in a self-cycle is not assigned to any object and is removed. If an object $s$ is part of a cycle, then $s$ is removed. Let $C_1=\{c_1^1,c_1^2,...,c_1^{L_1}\}$ be the set of cycles in Step $1$ (where $L_1 \geq 1$ is the number of cycles in Step $1$). Let $I_1$ be the set of agents who are assigned to some object at Step $1$. Let $S_1$ be the set of objects that are assigned to some agent at Step $1$. Let $m_1^l$ be all the matches from cycle $c_1^l$ that are formed in Step $1$ of the algorithm:
\begin{equation*}
m_1^l=\left\lbrace
\begin{matrix}
\{(i,s) \mid i,s \in c_1^l \text{ and } i\mapsto s \} & \text{if} & c_1^l \neq (j) \\
\{(j,j) \} & \text{if} & c_1^l = (j)
\end{matrix}\right.
\end{equation*}
where $(j,j)$ simply means that agent $j$ who is in a self-cycle ends up being definitely unassigned to any object. Let $M_1 = \cup_{l=1}^{L_1} m_1^l$ be all the matches between agents and objects formed in Step $1$ of the algorithm. Let $\bar{c}_1^l = \# \{i \in I \mid i \in c_1^l\}$ be the number of agents involved in cycle $c_1^l$, for $l=1,...,L_1$. Let $c_1^{\max} = \max \{\bar{c}_1^1,...,\bar{c}_1^{L_1}\}$.

\item[Step $k\geq 2$.] Each remaining agent $i \in I \setminus \cup_{l=1}^{k-1}I_l$ points to the object $s \in S \setminus \cup_{l=1}^{k-1}S_l$ that is ranked first in $P_i$. If there is no such object, then agent $i$ points to herself and she forms a self-cycle. Each object $s \in S \setminus \cup_{l=1}^{k-1}S_l$ points to the agent $j \in I \setminus \cup_{l=1}^{k-1}I_l$ that has the highest priority in $F_s$. There is at least one cycle. Every agent in a cycle is assigned to the object she points to and she is removed. Similarly, every agent in a self-cycle is not assigned to any object and is removed. If an object $s$ is part of a cycle, then $s$ is removed. Let $C_k=\{c_k^1,c_k^2,...,c_k^{L_k}\}$ be the set of cycles in Step $k$ (where $L_k \geq 1$ is the number of cycles in Step $k$). Let $I_k$ be the set of agents who are assigned to some object at Step $k$. Let $S_k$ be the set of objects that are assigned to some agent at Step $k$. Let $m_k^l$ be all the matches from cycle $c_k^l$ that are formed in Step $k$ of the algorithm:
\begin{equation*}
m_k^l=\left\lbrace
\begin{matrix}
\{(i,s) \mid i,s \in c_k^l \text{ and } i\mapsto s \} & \text{if} & c_k^l \neq (j) \\
\{(j,j) \} & \text{if} & c_k^l = (j)
\end{matrix}\right.
\end{equation*}
Let $M_k = \cup_{l=1}^{L_k} m_k^l$ be all the matches between agents and objects formed in Step $k$ of the algorithm. Let $\bar{c}_k^l = \# \{i \in I \mid i \in c_k^l\}$ be the number of agents involved in cycle $c_k^l$, for $l=1,...,L_k$. Let $c_k^{\max} = \max \{\bar{c}_k^1,...,\bar{c}_k^{L_k}\}$.

\item[End.] The algorithm stops when all agents have been removed. Let $\bar{k}$ be the step at which the algorithm stops. Let $\mu^T$ denote the matching obtained from the Top Trading Cycles mechanism and it is given by $\mu^T = \cup_{k=1}^{\bar{k}} M_k$. Let $\gamma = \max \{c_{1}^{\max},...,c_{\bar{k}}^{\max}\}$ be the maximum number of agents involved in any cycle of the TTC.
\end{itemize}

Abdulkadiro\u{g}lu and Sönmez (2003) show that the TTC mechanism is Pareto efficient and strategy-proof. TTC is also individually rational and non-wasteful, but it is not stable.

Another mechanism that is commonly adopted all over the world is Gale and Shapley's Deferred Acceptance (DA) algorithm. Let $\mu^D$ denote the matching obtained from the DA mechanism. The DA mechanism is strategy-proof and stable but not Pareto efficient.\footnote{Che and Tercieux (2019) show that both Pareto efficiency and stability can be achieved asymptotically using DA and TTC mechanisms when agents have uncorrelated preferences.} 

\section{Horizon-$k$ vNM Stable Set}\label{S4} 

Is it possible to stabilize the matching obtained from the TTC algorithm once agents become limited farsighted? If yes, how much farsightedness from the agents do we need? To answer this question we propose the notion of horizon-$k$ vNM stable set for priority-based matching problems to study the matchings that are stable when agents are limited in their degree of farsightedness. 

A horizon-$k$ improving path for priority-based matching problems is a sequence of matchings that can emerge when limited farsighted agents form or destroy matches based on the improvement the $k$-steps ahead matching offers them relative to the current matching. A set of matchings is a horizon-$k$ vNM stable set if (IS) for any two matchings belonging to the set, there is no horizon-$k$ improving path connecting from one matching to the other one, and (ES) there always exists a horizon-$k$ improving path from every matching outside the set to some matching within the set.

Given a matching $\mu \in \mathcal{M} $ with agent $i \in I $ assigned to object $s \in S, $ so $\mu(i) = s, $ the matching $\mu^{\prime} $ that is identical to $\mu, $ except that the match between $i$ and $s$ has been destroyed by either $i$ or $s$, is denoted by $\mu^{\prime} =\mu - (i,s)$. Given a matching $\mu \in \mathcal{M}$ such that $i \in I$ and $s \in S$ are not matched to one another, the matching $\mu^{\prime}$ that is identical to $\mu $, except that the pair $(i,s)$ has formed at $\mu^{\prime}$ (and some $j = \mu^{-1}(s)$ becomes unassigned if $\#\mu^{-1}(s)=1$), is denoted by $\mu^{\prime} = \mu + (i,s) - \{(i,\mu(i)) \mid \mu(i) \neq i\} - \{(\mu^{-1}(s),s) \mid \#(\mu^{-1}(s))=1\}$.

\begin{adefinition}\label{DMFLI} 
Let $\langle I,S,P,F\rangle$ be a priority-based matching problem. A horizon-$k$ improving path from a matching $\mu \in \mathcal{M}$ to a matching $\mu^{\prime } \in \mathcal{M} \setminus \{\mu \}$ is a finite sequence of distinct matchings $\mu_{0},\ldots ,\mu_{L}$ with $\mu_{0} = \mu$ and $\mu_{L} = \mu^{\prime }$ such that for every $l \in \{0,\ldots ,L-1\}$ either
\begin{itemize}
\item[(\textbf{i})]$ \mu_{l + 1} = \mu_{l} - (i,s)$ for some $(i,s) \in I \times S$ such that $\mu_{\min \{l+k,L\}}(i) P_i \mu _{l}(i)$, or
\item[(\textbf{ii})]$ \mu_{l + 1} = \mu_{l} + (i,s) - \{(i,\mu_{l}(i)) \mid \mu_{l}(i) \neq i\} - \{(\mu_{l}^{-1}(s),s) \mid \#(\mu_{l}^{-1}(s))=1\}$ for some $(i,s) \in I \times S$ such that $\mu_{\min \{l+k,L\}}(i) P_i \mu _{l}(i)$ and $F_s(i) < F_s(\mu_{l}^{-1}(s))$ if $\#(\mu_{l}^{-1}(s)) = 1$.
\end{itemize}
\end{adefinition}

Definition \ref{DMFLI} tells us that a horizon-$k$ improving path for priority-based matching problems consists of a sequence of matchings where along the sequence agents form or destroy matches based on the improvement the $k$-steps ahead matching offers them relative to the current one. Precisely, along a horizon-$k$ improving path, each time some agent $i$ is on the move she is comparing her current match (i.e. $\mu_l(i)$) with the match she will get $k$-steps ahead on the sequence (i.e. $\mu_{l+k}(i)$) except if the end matching of the sequence lies within her horizon (i.e. $L < l+k$). In such a case, she simply compares her current match (i.e. $\mu_l(i)$) with the end match (i.e. $\mu_L$). 

Objects can be assigned to any agent on their priority lists unless they have already been assigned to some agent. When an object $s \in S$ is already assigned to some agent $\mu_l^{-1}(s)$ at $\mu_l$, this object $s$ can be reassigned to another agent $\mu^{-1}_{l+1}(s) \neq \mu_l(s)$ at $\mu_{l+1}$ only if agent $\mu^{-1}_{l+1}(s)$ has a higher priority than agent $\mu_l^{-1}(s)$. 

Let some $\mu \in \mathcal{M}$ be given. If there exists a horizon-$k$ improving path from a matching $\mu$ to a matching $\mu^{\prime}$, then we write $\mu \rightarrow_k \mu^{\prime}$. The set of matchings $\mu^{\prime }\in \mathcal{M}$ such that there is a horizon-$k$ improving path from $\mu$ to $\mu^{\prime }$ is denoted by $\phi_k(\mu)$, so $\phi_k(\mu)=\{\mu^{\prime} \in \mathcal{M} \mid \mu \rightarrow_k \mu^{\prime} \}$.

\begin{adefinition} \label{D6} Let $\langle I,S,P,F\rangle$ be a priority-based matching problem. A set of matchings $V \subseteq \mathcal{M}$ is a horizon-$k$ vNM stable set if it satisfies:
\begin{itemize}
\item[(\textbf{i})] Internal stability (IS): For every $\mu, \mu^{\prime} \in V$,
it holds that $\mu^{\prime} \notin \phi_k(\mu).$
\item[(\textbf{ii})] External stability (ES): For every $\mu \in \mathcal{M} \setminus V$, it holds that $\phi_k(\mu) \cap V \neq \emptyset. $
\end{itemize}
\end{adefinition}

Condition (i) of Definition \ref{D6} corresponds to internal stability. For any two matchings $\mu $ and $\mu^{\prime }$ in the horizon-$k$ vNM stable set $V$ there is no horizon-$k$ improving path connecting $\mu$ to $\mu^{\prime }$. Condition (ii) of Definition \ref{D6} expresses external stability. There exists a horizon-$k$ improving path from every matching $\mu$ outside the horizon-$k$ vNM stable set $V$ to some matching in $V$.\footnote{Ehlers (2007) and Herings, Mauleon and Vannetelbosch (2017) study vNM stable sets when all agents are myopic in two-sided matching problems.}

\section{Main Result}\label{S5} 

Remember that $\gamma = \max \{c_{1}^{\max},...,c_{\bar{k}}^{\max}\}$ is the maximum number of agents involved in any cycle of the TTC.

\begin{atheorem}\label{T5}
Let $\langle I,S,P,F\rangle$ be a priority-based matching problem and $\mu^T$ be the matching obtained from the Top Trading Cycles mechanism. The singleton set $\{\mu^T\}$ is a horizon-$k$ stable set for $k\geq (3\gamma -1)$.
\end{atheorem}

\begin{proof}
Since $\{\mu^T\}$ is a singleton set, internal stability (IS) is satisfied. (ES) Take any matching $\mu \neq \mu^T$, we need to show that $\phi_k(\mu) \ni \mu^T$ for $k\geq (3\gamma -1)$. We build in steps a horizon-$k$ improving path from $\mu$ to $\mu^T$ for $k\geq (3\gamma -1)$.
\begin{itemize}
\item[Step 1.1.]If $m_1^1 \subseteq \mu$ and $1 \neq L_1$ then go to Step 1.2 with $\mu_{1,1}^{\prime \prime \prime} = \mu$. If $m_1^1 \subseteq \mu$ and $1 = L_1$ then go to Step 1.End with $\mu_{1,L_1}^{\prime \prime \prime} = \mu$. If $m_1^1 \nsubseteq \mu$ then go to Step 1.1.A.

\item[Step 1.1.A.]If $\{(i,s) \mid i,s \in c_1^1 \text{ and } s \mapsto i\} \subseteq \mu$ then go to Step 1.1.B with $\mu_{1,1}^{\prime} = \mu$. If $\{(i,s) \mid i,s \in c_1^1 \text{ and } s \mapsto i\} \nsubseteq \mu$ then there is some agent $i$ such that $s \neq \mu(i) \neq \mu^T(i)$ and $s \mapsto i$ with $i,s \in c_1^1$. This agent $i$ matches with object $s$ that ranks her first on its priority list. We reach the matching $\mu_{1,1,1}^{\prime}  = \mu + (i,s) - \{(i,\mu(i)) \mid \mu(i) \neq i\} - \{(\mu^{-1}(s),s) \mid \#(\mu^{-1}(s))=1\}$ where $s \mapsto i$ and $i,s \in c_1^1$. If $\{(i,s) \mid i,s \in c_1^1 \text{ and } s \mapsto i\} \subseteq \mu_{1,1,1}^{\prime} $ then go to Step 1.1.B with $\mu_{1,1}^{\prime} = \mu_{1,1,1}^{\prime} $. If $\{(i,s) \mid i,s \in c_1^1 \text{ and } s \mapsto i\} \nsubseteq \mu_{1,1,1}^{\prime} $ then there is some agent $i^{\prime}$ such that $s^{\prime} \neq \mu_{1,1,1}^{\prime}(i^{\prime}) \neq \mu^T(i^{\prime})$ and $s^{\prime} \mapsto i^{\prime}$ with $i^{\prime},s^{\prime} \in c_1^1$. This agent $i^{\prime}$ matches with object $s^{\prime}$ that ranks her first on its priority list. We reach the matching $\mu_{1,1,2}^{\prime}  = \mu_{1,1,1}^{\prime}  + (i^{\prime},s^{\prime}) - \{(i^{\prime},\mu_{1,1,1}^{\prime}(i^{\prime})) \mid \mu_{1,1,1}^{\prime}(i^{\prime})  \neq i^{\prime}\} - \{(\mu_{1,1,1}^{\prime -1}(s^{\prime}),s^{\prime}) \mid \#(\mu_{1,1,1}^{\prime -1}(s^{\prime}))=1\}$ where $s^{\prime} \mapsto i^{\prime}$ and $i^{\prime},s^{\prime} \in c_1^1$. We proceed as above until we reach the matching $\mu_{1,1}^{\prime } = \mu + \{(i,s) \mid i,s \in c_1^1 \text{ and } s \mapsto i\} - \{(i,\mu(i)) \mid i,s \in c_1^1 \text{, } s \mapsto i \text{ and } \mu(i) \neq s\} - \{(\mu^{-1}(s),s) \mid i,s \in c_1^1 \text{, } s \mapsto i \text{ and } \mu^{-1}(s) \neq i\}$ where each agent involved in $c_1^1$ is matched to the object that ranks her first on its priority list. Step 1.1.A counts at most $\bar{c}^1_1$ steps.

\item[Step 1.1.B.]Let $\mathcal{I}_1^1=\{(i^l)\}_{l=1}^{\bar{c}_1^1}$ be such that $i^l \in c_1^1$ and $i^l = i_{o_l} \neq i^{l+1} = i_{o_{l+1}}$ with $o_{l} < o_{l+1}$ for $l=1,...,\bar{c}_1^1-1$. That is, $\mathcal{I}_1^1$ is an ordered set of the agents involved in cycle $c_1^1$ where $\bar{c}_1^1=\#\{i \in I \mid i \in c_1^1\}$ is the number of agents involved in cycle $c_1^1$. From the matching $\mu_{1,1}^{\prime }$, agents $i^1$ to $i^{\bar{c}_1^1-1}$ successively leave their objects to reach the matching $\mu_{1,1}^{\prime \prime} = \mu_{1,1}^{\prime} - \{(i,s) \mid i,s \in c_1^1 \text{, } s \mapsto i \text{ and } i \neq i^{\bar{c}_1^1}\} $ where only agent $i^{\bar{c}_1^1}$ is still matched to the object that ranks her first on its priority list. Step 1.1.B counts at most $\bar{c}^1_1 -1$ steps.

\item[Step 1.1.C.]From the matching $\mu_{1,1}^{\prime \prime}$, agent $i^{\bar{c}_1^1}$ first matches with her top choice object $s$ to reach the matching $\mu_{1,1}^{\prime \prime} + (i^{\bar{c}_1^1},s) - (i^{\bar{c}_1^1}, \mu_{1,1}^{\prime \prime}(i^{\bar{c}_1^1}))$ where $s$ is such that $i \mapsto s \text{ and } i,s \in c^1_1$. Notice that $s$ was unassigned at $\mu_{1,1}^{\prime \prime}$ while the object that ranks $i^{\bar{c}_1^1}$ first on its priority list, i.e. $\mu_{1,1}^{\prime \prime}(i^{\bar{c}_1^1})$, is now unassigned. Next, agents $i^1$ to $i^{\bar{c}_1^1-1}$ successively match to their top choice object to reach the matching $\mu_{1,1}^{\prime \prime \prime} = \mu_{1,1}^{\prime \prime} - \{(i,s) \mid i,s \in c_1^1 \text{ and } s \mapsto i \} + \{(i,s) \mid i,s \in c_1^1 \text{ and } i \mapsto s\}$. Step 1.1.C counts at most $\bar{c}^1_1$ steps. We have reached $\mu_{1,1}^{\prime \prime \prime}$ with $m_1^1 \subseteq \mu_{1,1}^{\prime \prime \prime}$ and so agents belonging to $c_1^1$ are assigned to the same object as in $\mu^T$. Step 1.1 counts at most $3 \bar{c}^1_1 - 1$ steps. Hence, it is sufficient that the agents who move in Step 1.1 look forward $3 \bar{c}^1_1 - 1$ steps ahead to have incentives for engaging the move towards the matching $\mu_{1,1}^{\prime \prime \prime}$ where they already get the object assigned by the TTC. Once they reach those matches they do not move afterwards. If $1 \neq L_1$, then go to Step 1.2. Otherwise, go to Step 1.End with $\mu_{1,L_1}^{\prime \prime \prime} = \mu_{1,1}^{\prime \prime \prime}$.

\item[Step 1.$k$.]($k>1$) If $m_1^k \subseteq \mu_{1,k-1}^{\prime \prime \prime}$ and $k \neq L_1$ then go to Step 1.k+1 with $\mu_{1,k}^{\prime \prime \prime} = \mu_{1,k-1}^{\prime \prime \prime}$. If $m_1^k \subseteq \mu_{1,k-1}^{\prime \prime \prime}$ and $k = L_1$ then go to Step 1.End with $\mu_{1,L_1}^{\prime \prime \prime} = \mu_{1,k-1}^{\prime \prime \prime}$. If $m_1^k \nsubseteq \mu_{1,k-1}^{\prime \prime \prime}$ then go to Step 1.$k$.A.

\item[Step 1.$k$.A.]If $\{(i,s) \mid i,s \in c_1^{k} \text{ and } s \mapsto i\} \subseteq \mu_{1,k-1}^{\prime \prime \prime}$ then go to Step 1.$k$.B with $\mu_{1,k}^{\prime} = \mu_{1,k-1}^{\prime \prime \prime}$. If $\{(i,s) \mid i,s \in c_1^k \text{ and } s \mapsto i\} \nsubseteq \mu_{1,k-1}^{\prime \prime \prime}$ then there is some agent $i$ such that $s \neq \mu_{1,k-1}^{\prime \prime \prime}(i) \neq \mu^T(i)$ and $s \mapsto i$ with $i,s \in c_1^k$. This agent $i$ matches with object $s$ that ranks her first on its priority list. We reach the matching $\mu_{1,k,1}^{\prime}  = \mu_{1,k-1}^{\prime \prime \prime} + (i,s) - \{(i,\mu_{1,k-1}^{\prime \prime \prime}(i)) \mid \mu_{1,k-1}^{\prime \prime \prime}(i) \neq i\} - \{(\mu_{1,k-1}^{\prime \prime \prime -1}(s),s) \mid \#(\mu_{1,k-1}^{\prime \prime \prime -1}(s))=1\}$ where $s \mapsto i$ and $i,s \in c_1^k$. If $\{(i,s) \mid i,s \in c_1^k \text{ and } s \mapsto i\} \subseteq \mu_{1,k,1}^{\prime} $ then go to Step 1.$k$.B with $\mu_{1,k}^{\prime} = \mu_{1,k,1}^{\prime} $. If $\{(i,s) \mid i,s \in c_1^k \text{ and } s \mapsto i\} \nsubseteq \mu_{1,k,1}^{\prime} $ then there is some agent $i^{\prime}$ such that $s^{\prime} \neq \mu_{1,k,1}^{\prime}(i^{\prime}) \neq \mu^T(i^{\prime})$ and $s^{\prime} \mapsto i^{\prime}$ with $i^{\prime},s^{\prime} \in c_1^k$. This agent $i^{\prime}$ matches with object $s^{\prime}$ that ranks her first on its priority list. We reach the matching $\mu_{1,k,2}^{\prime}  = \mu_{1,k,1}^{\prime}  + (i^{\prime},s^{\prime}) - \{(i^{\prime},\mu_{1,k,1}^{\prime}(i^{\prime})) \mid \mu_{1,k,1}^{\prime}(i^{\prime})  \neq i^{\prime}\} - \{(\mu_{1,k,1}^{\prime -1}(s^{\prime}),s^{\prime}) \mid \#(\mu_{1,k,1}^{\prime -1}(s^{\prime}))=1\}$ where $s^{\prime} \mapsto i^{\prime}$ and $i^{\prime},s^{\prime} \in c_1^k$. We proceed as above until we reach the matching $\mu_{1,k}^{\prime } = \mu_{1,k-1}^{\prime \prime \prime} + \{(i,s) \mid i,s \in c_1^k \text{ and } s \mapsto i\} - \{(i,\mu_{1,k-1}^{\prime \prime \prime}(i)) \mid i,s \in c_1^k \text{, } s \mapsto i \text{ and } \mu_{1,k-1}^{\prime \prime \prime}(i) \neq s\} - \{(\mu_{1,k-1}^{\prime \prime \prime -1}(s),s) \mid i,s \in c_1^k \text{, } s \mapsto i \text{ and } \mu_{1,k-1}^{\prime \prime \prime -1}(s) \neq i\}$ where each agent involved in $c_1^k$ is matched to the object that ranks her first on its priority list. Step 1.$k$.A counts at most $\bar{c}^k_1$ steps.

\item[Step 1.$k$.B.]Let $\mathcal{I}_1^k=\{(i^l)\}_{l=1}^{\bar{c}_1^k}$ be such that $i^l \in c_1^k$ and $i^l = i_{o_l} \neq i^{l+1} = i_{o_{l+1}}$ with $o_{l} < o_{l+1}$ for $l=1,...,\bar{c}_1^k-1$. That is, $\mathcal{I}_1^k$ is an ordered set of the agents involved in cycle $c_1^k$ where $\bar{c}_1^k=\#\{i \in I \mid i \in c_1^k\}$ is the number of agents involved in cycle $c_1^k$. From the matching $\mu_{1,k}^{\prime }$, agents $i^1$ to $i^{\bar{c}_1^k-1}$ successively leave their objects to reach the matching $\mu_{1,k}^{\prime \prime} = \mu_{1,k}^{\prime} - \{(i,s) \mid i,s \in c_1^k \text{, } s \mapsto i \text{ and } i \neq i^{\bar{c}_1^k}\} $ where only agent $i^{\bar{c}_1^k}$ is still matched to the object that ranks her first on its priority list. Step 1.$k$.B counts at most $\bar{c}^k_1 -1$ steps.

\item[Step 1.$k$.C.]From the matching $\mu_{1,k}^{\prime \prime}$, agent $i^{\bar{c}_1^k}$ first matches with her top choice object $s$ to reach the matching $\mu_{1,k}^{\prime \prime} + (i^{\bar{c}_1^k},s) - (i^{\bar{c}_1^k}, \mu_{1,k}^{\prime \prime}(i^{\bar{c}_1^k}))$ where $s$ is such that $i \mapsto s \text{ and } i,s \in c^k_1$. Notice that $s$ was unassigned at $\mu_{1,k}^{\prime \prime}$ while the object that ranks $i^{\bar{c}_1^k}$ first on its priority list, i.e. $\mu_{1,k}^{\prime \prime}(i^{\bar{c}_1^k})$, is now unassigned. Next, agents $i^1$ to $i^{\bar{c}_1^k-1}$ successively match to their top choice object to reach the matching $\mu_{1,k}^{\prime \prime \prime} = \mu_{1,k}^{\prime \prime} - \{(i,s) \mid i,s \in c_1^k \text{ and } s \mapsto i \} + \{(i,s) \mid i,s \in c_1^k \text{ and } i \mapsto s\}$. Step 1.$k$.C counts at most $\bar{c}^k_1$ steps. We have reached $\mu_{1,k}^{\prime \prime \prime}$ with $m_1^k \subseteq \mu_{1,k}^{\prime \prime \prime}$ and so agents belonging to $c_1^k$ are assigned to the same object as in $\mu^T$. Step 1.$k$ counts at most $3 \bar{c}^k_1 - 1$ steps. Hence, it is sufficient that the agents who move in Step 1.$k$ look forward $3 \bar{c}^k_1 - 1$ steps ahead to have incentives for engaging the move towards the matching $\mu_{1,k}^{\prime \prime \prime}$ where they already get the object assigned by the TTC. Once they reach those matches they do not move afterwards. If $1 \neq L_1$, then go to Step 1.$k+1$. Otherwise, go to Step 1.End with $\mu_{1,L_1}^{\prime \prime \prime} = \mu_{1,k}^{\prime \prime \prime}$.

\item[Step 1.End.] We have reached $\mu_{1,L_1}^{\prime \prime \prime}$ with $\cup_{l=1}^{L_1} m_1^{l} = M_1 \subseteq \mu_{1,L_1}^{\prime \prime \prime}$. If $\mu_{1,L_1}^{\prime \prime \prime} = \mu^T$ then the process ends. Otherwise, go to Step 2.1.

\item[Step 2.1.]If $m_2^1 \subseteq \mu_{1,L_1}^{\prime \prime \prime}$ and $1 \neq L_2$ then go to Step 2.2 with $\mu_{2,1}^{\prime \prime \prime} = \mu_{1,L_1}^{\prime \prime \prime}$. If $m_2^1 \subseteq \mu_{1,L_1}^{\prime \prime \prime}$ and $1 = L_2$ then go to Step 2.End with $\mu_{2,L_2}^{\prime \prime \prime} = \mu_{1,L_1}^{\prime \prime \prime}$. If $m_2^1 \nsubseteq \mu_{1,L_1}^{\prime \prime \prime}$ then go to Step 2.1.A.

\item[Step 2.1.A.]If $\{(i,s) \mid i,s \in c_2^1 \text{ and } s \mapsto i\} \subseteq \mu_{1,L_1}^{\prime \prime \prime}$ then go to Step 2.1.B with $\mu_{2,1}^{\prime} = \mu_{1,L_1}^{\prime \prime \prime}$. If $\{(i,s) \mid i,s \in c_2^1 \text{ and } s \mapsto i\} \nsubseteq \mu_{1,L_1}^{\prime \prime \prime}$ then there is some agent $i$ such that $s \neq \mu_{1,L_1}^{\prime \prime \prime}(i) \neq \mu^T(i)$ and $s \mapsto i$ with $i,s \in c_2^1$. This agent $i$ matches with object $s \in S \setminus S_1$ that ranks her first on its priority list among agents belonging to $I \setminus I_1$. We reach the matching $\mu_{2,1,1}^{\prime}  = \mu_{1,L_1}^{\prime \prime \prime} + (i,s) - \{(i,\mu_{1,L_1}^{\prime \prime \prime}(i)) \mid \mu_{1,L_1}^{\prime \prime \prime}(i) \neq i\} - \{(\mu_{1,L_1}^{\prime \prime \prime -1}(s),s) \mid \#(\mu_{1,L_1}^{\prime \prime \prime -1}(s))=1\}$ where $s \mapsto i$ and $i,s \in c_2^1$. If $\{(i,s) \mid i,s \in c_2^1 \text{ and } s \mapsto i\} \subseteq \mu_{2,1,1}^{\prime} $ then go to Step 2.1.B with $\mu_{2,1}^{\prime} = \mu_{2,1,1}^{\prime} $. If $\{(i,s) \mid i,s \in c_2^1 \text{ and } s \mapsto i\} \nsubseteq \mu_{2,1,1}^{\prime} $ then there is some agent $i^{\prime}$ such that $s^{\prime} \neq \mu_{2,1,1}^{\prime}(i^{\prime}) \neq \mu^T(i^{\prime})$ and $s^{\prime} \mapsto i^{\prime}$ with $i^{\prime},s^{\prime} \in c_2^1$. This agent $i^{\prime}$ matches with object $s^{\prime} \in S \setminus S_1$ that ranks her first on its priority list among agents belonging to $I \setminus I_1$. We reach the matching $\mu_{2,1,2}^{\prime}  = \mu_{2,1,1}^{\prime}  + (i^{\prime},s^{\prime}) - \{(i^{\prime},\mu_{2,1,1}^{\prime}(i^{\prime})) \mid \mu_{2,1,1}^{\prime}(i^{\prime})  \neq i^{\prime}\} - \{(\mu_{2,1,1}^{\prime -1}(s^{\prime}),s^{\prime}) \mid \#(\mu_{2,1,1}^{\prime -1}(s^{\prime}))=1\}$ where $s^{\prime} \mapsto i^{\prime}$ and $i^{\prime},s^{\prime} \in c_2^1$. We proceed as above until we reach the matching $\mu_{2,1}^{\prime } = \mu_{1,L_1}^{\prime \prime \prime} + \{(i,s) \mid i,s \in c_2^1 \text{ and } s \mapsto i\} - \{(i,\mu_{1,L_1}^{\prime \prime \prime}(i)) \mid i,s \in c_2^1 \text{, } s \mapsto i \text{ and } \mu_{1,L_1}^{\prime \prime \prime}(i) \neq s\} - \{(\mu_{1,L_1}^{\prime \prime \prime -1}(s),s) \mid i,s \in c_2^1 \text{, } s \mapsto i \text{ and } \mu_{1,L_1}^{\prime \prime \prime -1}(s) \neq i\}$ where each agent involved in $c_2^1$ is matched to the object that ranks her first on its priority list among agents belonging to $I \setminus I_1$. Step 2.1.A counts at most $\bar{c}^2_1$ steps.

\item[Step 2.1.B.]Let $\mathcal{I}_2^1=\{(i^l)\}_{l=1}^{\bar{c}_2^1}$ be such that $i^l \in c_2^1$ and $i^l = i_{o_l} \neq i^{l+1} = i_{o_{l+1}}$ with $o_{l} < o_{l+1}$ for $l=1,...,\bar{c}_2^1-1$. That is, $\mathcal{I}_2^1$ is an ordered set of the agents involved in cycle $c_2^1$ where $\bar{c}_2^1=\#\{i \in I \mid i \in c_2^1\}$ is the number of agents involved in cycle $c_2^1$. From the matching $\mu_{2,1}^{\prime }$, agents $i^1$ to $i^{\bar{c}_2^1-1}$ successively leave their objects to reach the matching $\mu_{2,1}^{\prime \prime} = \mu_{2,1}^{\prime} - \{(i,s) \mid i,s \in c_2^1 \text{, } s \mapsto i \text{ and } i \neq i^{\bar{c}_2^1}\} $ where only agent $i^{\bar{c}_2^1}$ is still matched to the object that ranks her first on its priority list among agents belonging to $I \setminus I_1$. Step 2.1.B counts at most $\bar{c}^2_1 -1$ steps.

\item[Step 2.1.C.]From the matching $\mu_{2,1}^{\prime \prime}$, agent $i^{\bar{c}_2^1}$ first matches with her top choice object $s \in S \setminus S_1$ to reach the matching $\mu_{2,1}^{\prime \prime} + (i^{\bar{c}_2^1},s) - (i^{\bar{c}_2^1}, \mu_{2,1}^{\prime \prime}(i^{\bar{c}_2^1}))$ where $s$ is such that $i \mapsto s \text{ and } i,s \in c^2_1$. Notice that $s$ was unassigned at $\mu_{2,1}^{\prime \prime}$ while the object that ranks $i^{\bar{c}_2^1}$ first on its priority list among agents belonging to $I \setminus I_1$, i.e. $\mu_{2,1}^{\prime \prime}(i^{\bar{c}_2^1})$, is now unassigned. Next, agents $i^1$ to $i^{\bar{c}_2^1-1}$ successively match to their top choice object in $S \setminus S_1$ to reach the matching $\mu_{2,1}^{\prime \prime \prime} = \mu_{2,1}^{\prime \prime} - \{(i,s) \mid i,s \in c_2^1 \text{ and } s \mapsto i \} + \{(i,s) \mid i,s \in c_2^1 \text{ and } i \mapsto s\}$. Step 2.1.C counts at most $\bar{c}^1_2$ steps. We have reached $\mu_{2,1}^{\prime \prime \prime}$ with $m_2^1 \subseteq \mu_{2,1}^{\prime \prime \prime}$ and so agents belonging to $c_2^1$ are assigned to the same object as in $\mu^T$. Step 2.1 counts at most $3 \bar{c}^1_2 - 1$ steps. Hence, it is sufficient that the agents who move in Step 2.1 look forward $3 \bar{c}^1_2 - 1$ steps ahead to have incentives for engaging the move towards the matching $\mu_{2,1}^{\prime \prime \prime}$ where they already get the object assigned by the TTC. Once they reach those matches they do not move afterwards. If $1 \neq L_2$, then go to Step 2.2. Otherwise, go to Step 2.End with $\mu_{2,L_2}^{\prime \prime \prime} = \mu_{2,1}^{\prime \prime \prime}$.

\item[Step 2.$k$.]($k>1$) If $m_2^k \subseteq \mu_{2,k-1}^{\prime \prime \prime}$ and $k \neq L_2$ then go to Step 2.k+1 with $\mu_{2,k}^{\prime \prime \prime} = \mu_{2,k-1}^{\prime \prime \prime}$. If $m_2^k \subseteq \mu_{2,k-1}^{\prime \prime \prime}$ and $k = L_2$ then go to Step 2.End with $\mu_{2,L_2}^{\prime \prime \prime} = \mu_{2,k-1}^{\prime \prime \prime}$. If $m_2^k \nsubseteq \mu_{2,k-1}^{\prime \prime \prime}$ then go to Step 2.$k$.A.

\item[Step 2.$k$.A]If $\{(i,s) \mid i,s \in c_2^{k} \text{ and } s \mapsto i\} \subseteq \mu_{2,k-1}^{\prime \prime \prime}$ then go to Step 2.$k$.B with $\mu_{2,k}^{\prime} = \mu_{2,k-1}^{\prime \prime \prime}$. If $\{(i,s) \mid i,s \in c_2^k \text{ and } s \mapsto i\} \nsubseteq \mu_{2,k-1}^{\prime \prime \prime}$ then there is some agent $i$ such that $s \neq \mu_{2,k-1}^{\prime \prime \prime}(i) \neq \mu^T(i)$ and $s \mapsto i$ with $i,s \in c_2^k$. This agent $i$ matches with object $s$ that ranks her first on its priority list among agents belonging to $I \setminus I_1$. We reach the matching $\mu_{2,k,1}^{\prime}  = \mu_{2,k-1}^{\prime \prime \prime} + (i,s) - \{(i,\mu_{2,k-1}^{\prime \prime \prime}(i)) \mid \mu_{2,k-1}^{\prime \prime \prime}(i) \neq i\} - \{(\mu_{2,k-1}^{\prime \prime \prime -1}(s),s) \mid \#(\mu_{2,k-1}^{\prime \prime \prime -1}(s))=1\}$ where $s \mapsto i$ and $i,s \in c_2^k$. If $\{(i,s) \mid i,s \in c_2^k \text{ and } s \mapsto i\} \subseteq \mu_{2,k,1}^{\prime} $ then go to Step 2.$k$.B with $\mu_{2,k}^{\prime} = \mu_{2,k,1}^{\prime} $. If $\{(i,s) \mid i,s \in c_2^k \text{ and } s \mapsto i\} \nsubseteq \mu_{2,k,1}^{\prime} $ then there is some agent $i^{\prime}$ such that $s^{\prime} \neq \mu_{2,k,1}^{\prime}(i^{\prime}) \neq \mu^T(i^{\prime})$ and $s^{\prime} \mapsto i^{\prime}$ with $i^{\prime},s^{\prime} \in c_2^k$. This agent $i^{\prime}$ matches with object $s^{\prime}$ that ranks her first on its priority list among agents belonging to $I \setminus I_1$. We reach the matching $\mu_{2,k,2}^{\prime}  = \mu_{2,k,1}^{\prime}  + (i^{\prime},s^{\prime}) - \{(i^{\prime},\mu_{2,k,1}^{\prime}(i^{\prime})) \mid \mu_{2,k,1}^{\prime}(i^{\prime})  \neq i^{\prime}\} - \{(\mu_{2,k,1}^{\prime -1}(s^{\prime}),s^{\prime}) \mid \#(\mu_{2,k,1}^{\prime -1}(s^{\prime}))=1\}$ where $s^{\prime} \mapsto i^{\prime}$ and $i^{\prime},s^{\prime} \in c_2^k$. We proceed as above until we reach the matching $\mu_{2,k}^{\prime } = \mu_{2,k-1}^{\prime \prime \prime} + \{(i,s) \mid i,s \in c_2^k \text{ and } s \mapsto i\} - \{(i,\mu_{2,k-1}^{\prime \prime \prime}(i)) \mid i,s \in c_2^k \text{, } s \mapsto i \text{ and } \mu_{2,k-1}^{\prime \prime \prime}(i) \neq s\} - \{(\mu_{2,k-1}^{\prime \prime \prime -1}(s),s) \mid i,s \in c_2^k \text{, } s \mapsto i \text{ and } \mu_{2,k-1}^{\prime \prime \prime -1}(s) \neq i\}$ where each agent involved in $c_2^k$ is matched to the object that ranks her first on its priority list among agents belonging to $I \setminus I_1$. Step 2.$k$.A counts at most $\bar{c}^k_2$ steps.

\item[Step 2.$k$.B.]Let $\mathcal{I}_2^k=\{(i^l)\}_{l=1}^{\bar{c}_2^k}$ be such that $i^l \in c_2^k$ and $i^l = i_{o_l} \neq i^{l+1} = i_{o_{l+1}}$ with $o_{l} < o_{l+1}$ for $l=1,...,\bar{c}_2^k-1$. That is, $\mathcal{I}_2^k$ is an ordered set of the agents involved in cycle $c_2^k$ where $\bar{c}_2^k=\#\{i \in I \mid i \in c_2^k\}$ is the number of agents involved in cycle $c_2^k$. From the matching $\mu_{2,k}^{\prime }$, agents $i^1$ to $i^{\bar{c}_2^k-1}$ successively leave their objects to reach the matching $\mu_{2,k}^{\prime \prime} = \mu_{2,k}^{\prime} - \{(i,s) \mid i,s \in c_2^k \text{, } s \mapsto i \text{ and } i \neq i^{\bar{c}_2^k}\} $ where only agent $i^{\bar{c}_2^k}$ is still matched to the object that ranks her first on its priority list among agents belonging to $I \setminus I_1$. Step 2.$k$.B counts at most $\bar{c}^k_2 -1$ steps.

\item[Step 2.$k$.C.]From the matching $\mu_{2,k}^{\prime \prime}$, agent $i^{\bar{c}_2^k}$ first matches with her top choice object $s \in S \setminus S_1$ to reach the matching $\mu_{2,k}^{\prime \prime} + (i^{\bar{c}_2^k},s) - (i^{\bar{c}_2^k}, \mu_{2,k}^{\prime \prime}(i^{\bar{c}_2^k}))$ where $s$ is such that $i \mapsto s \text{ and } i,s \in c^k_2$. Notice that $s$ was unassigned at $\mu_{2,k}^{\prime \prime}$ while the object that ranks $i^{\bar{c}_2^k}$ first on its priority list among agents belonging to $I \setminus I_1$, i.e. $\mu_{2,k}^{\prime \prime}(i^{\bar{c}_2^k})$, is now unassigned. Next, agents $i^1$ to $i^{\bar{c}_2^k-1}$ successively match to their top choice object in $S \setminus S_1$ to reach the matching $\mu_{2,k}^{\prime \prime \prime} = \mu_{2,k}^{\prime \prime} - \{(i,s) \mid i,s \in c_2^k \text{ and } s \mapsto i \} + \{(i,s) \mid i,s \in c_2^k \text{ and } i \mapsto s\}$. Step 2.$k$.C counts at most $\bar{c}^k_2$ steps. We have reached $\mu_{2,k}^{\prime \prime \prime}$ with $m_2^k \subseteq \mu_{2,k}^{\prime \prime \prime}$ and so agents belonging to $c_2^k$ are assigned to the same object as in $\mu^T$. Step 2.$k$ counts at most $3 \bar{c}^k_2 - 1$ steps. Hence, it is sufficient that the agents who move in Step 2.$k$ look forward $3 \bar{c}^k_2 - 1$ steps ahead to have incentives for engaging the move towards the matching $\mu_{2,k}^{\prime \prime \prime}$ where they already get the object assigned by the TTC. Once they reach those matches they do not move afterwards. If $1 \neq L_2$, then go to Step 2.$k+1$. Otherwise, go to Step 2.End with $\mu_{2,L_2}^{\prime \prime \prime} = \mu_{2,k}^{\prime \prime \prime}$.

\item[Step 2.End.] We have reached $\mu_{2,L_2}^{\prime \prime \prime}$ with $M_1 \cup M_2 \subseteq \mu_{2,L_2}^{\prime \prime \prime}$. If $\mu_{2,L_2}^{\prime \prime \prime} = \mu^T$ then the process ends. Otherwise, go to Step 3.1.

\item[End.]The process goes on until we reach $\mu_{\bar{k},L_{\bar{k}}}^{\prime \prime \prime} = \cup_{k=1}^{\bar{k}} M_k = \mu^T$.

Given $k\geq 3 \gamma -1$, we have that, along the constructed horizon-$k$ improving path, each time an agent $i$ is on the move she has incentives to do so since her end match (i.e. her TTC match $\mu^T(i)$) is within her horizon.

\end{itemize}
\end{proof}

The matching obtained from the TTC algorithm is always Pareto efficient but may not be stable when agents are myopic. Theorem \ref{T5} shows that, once agents are sufficiently farsighted (i.e. $k\geq 3 \gamma -1$), the matching obtained from the TTC algorithm becomes stable. Example \ref{Ex1} highlights Theorem \ref{T5}. In addition, it shows that, once agents are no more myopic, the matching obtained from the Deferred Acceptance (DA) algorithm may become unstable.
 
\begin{aexample}\label{Ex1} 
Consider a priority-based matching problem $\langle I,S,P,F\rangle$ with $I=\{i_1,i_2,i_3\}$ and $S=\{s_1,s_2,s_3\}$. Agents' preferences and objects' priorities are as follows.
\begin{center}
\begin{tabular}{ccc}
\multicolumn{3}{c}{Agents}\\
\hline
$P_{i_1}$& $P_{i_2}$ & $P_{i_3}$  \\
\hline
$s_1$ & $s_1$ & $s_2$   \\
$s_3$ & $s_2$ & $s_1$   \\
$s_2$ & $s_3$ & $s_3$   
\end{tabular}
~\qquad~
\begin{tabular}{ccc}
\multicolumn{3}{c}{Objects}\\
\hline
$F_{s_1}$& $F_{s_2}$ & $F_{s_3}$\\
\hline
$i_3$ & $i_2$ & $i_2$\\
$i_1$ & $i_1$ & $i_3$\\
$i_2$ & $i_3$ & $i_1$
    \end{tabular}
  \end{center}

\end{aexample}

Using Example \ref{Ex1} we provide the basic intuition behind Theorem \ref{T5} and its proof. In Example \ref{Ex1}, $\mu^T=\{(i_1,s_3),(i_2,s_1),(i_3,s_2)\}$ is the matching obtained from the TTC algorithm. In the first round of the TTC algorithm, there is one cycle where agent $i_2$ points to object $s_1$, object $s_1$ points to agent $i_3$, agent $i_3$ points to object $s_2$ and object $s_2$ points to agent $i_2$. That is, $C_1=\{c_1^1\}$ with $c_1^1=(s_1,i_3,s_2,i_2)$. Agent $i_2$ is assigned to object $s_1$ and agent $i_3$ is assigned to object $s_2$: $m_1^1=\{(i_2,s_1),(i_3,s_2)\}$, and so $i_2$ and $i_3$ exchange their priority. In the second round of the TTC algorithm, there is only one leftover agent, $i_1$, who points to object $s_3$ and one leftover object, $s_3$, that points to agent $i_1$. That is, $C_2=\{c_2^1\}$ with $c_2^1=(s_3,i_1)$. Agent $i_1$ is assigned to object $s_3$: $m_2^1=\{(i_1,s_3)\}$, and so $\mu^T=m_1^1 \cup m_2^1$. 

From Theorem \ref{T5} we know that $\{\mu^T\}$ is a horizon-$k$ vNM stable set for $k\geq 3 \gamma -1$. Indeed, if $k\geq 3 \gamma -1$, then from any $\mu \neq \mu^T$ there exists a horizon-$k$ improving path leading to $\mu^T$. In Example \ref{Ex1}, $\gamma =c_1^{\max}=2$. Take for instance the matching $\mu_0 = \{(i_1,s_1),(i_2,s_2),(i_3,s_3)\}$. We now construct a horizon-$k$ improving from $\mu_0$  to $\mu^T = \{(i_1,s_3),(i_2,s_1),(i_3,s_2)\} = \mu_5$ following the steps as in the proof of Theorem \ref{T5}. First, we consider agents and objects belonging to the cycles in $C_1$. Notice that $m_1^1 =\{(i_2,s_1),(i_3,s_2)\} \cap  \mu_0 = \emptyset$. Looking forward towards $\mu_4$ and $\mu^T$ (where $\mu_4(i_3)=\mu^T(i_3)$), agent $i_3$ matches to the object $s_1$ that ranks her first on its priority list to reach the matching $\mu_1 = \{(i_1,i_1),(i_2,s_2),(i_3,s_1)\}$ where agent $i_3$ is matched to the object in $c_1^1$ where she has priority. By doing so, agent $i_1$ is left without object. In $\mu_0$ (and $\mu_1$), agent $i_2$ is already matched to the object in $c_1^1$ where she has priority.\footnote{Hence, it is sufficient that agent $i_3$ looks forward at least $4$ (instead of $5$) steps ahead for having incentives to engage her first move towards $\mu_4$.}  Next, agent $i_3$ leaves her object $s_1$ to reach the matching $\mu_2 = \{(i_1,i_1),(i_2,s_2),(i_3,i_3)\}$ where agent $i_3$ is not assigned to any object. Agent $i_3$ is temporarily worse off, but she anticipates to end up being matched with $\mu^T(i_3)$. Next, agent $i_2$ matches to $s_1$ that was left unassigned to reach the matching $\mu_3 = \{(i_1,i_1),(i_2,s_1),(i_3,i_3)\}$. Next, agent $i_3$ matches to $s_2$ that was left by $i_2$ to reach the matching $\mu_4 = \{(i_1,i_1),(i_2,s_1),(i_3,s_2)\}$ with $m_1^1 = \{(i_2,s_1),(i_3,s_2)\} \subseteq \mu_4$. Finally, we consider agents and objects belonging to the cycles in $C_2$. Since $m_2^1 =\{(i_1,s_3)\} \cap  \mu_4 = \emptyset$, agent $i_1$ is assigned to object $s_3$ to form the match $(i_1,s_3)$ and to reach the matching $\mu_5 = \mu^T$. Thus, $\mu^T \in \phi_k(\mu_0)$.

In Example \ref{Ex1}, $\mu^D=\{(i_1,s_3),(i_2,s_2),(i_3,s_1)\}$ is the matching obtained from the Deferred Acceptance (DA) algorithm, $\mu^{B}=\{(i_1,s_1),(i_2,s_3),(i_3,s_2)\}$ is the matching obtained from the Immediate Acceptance (IA) algorithm (i.e. the Boston mechanism). Thus, $\mu^T \neq \mu^D \neq \mu^B$.

Let
\begin{equation*}
\begin{split}
\mu^1 & = \{(i_1,s_1),(i_2,i_2),(i_3,s_2)\}, \\
\mu^2 & = \{(i_1,s_1),(i_2,s_3),(i_3,s_2)\} = \mu^B, \\
\mu^3 & = \{(i_1,s_1),(i_2,s_2),(i_3,i_3)\}, \\
\mu^4 & = \{(i_1,s_1),(i_2,s_2),(i_3,s_3)\}, \\
\mu^5 & = \{(i_1,i_1),(i_2,s_2),(i_3,s_1)\}.
\end{split}
\end{equation*}

Since agent $i_1$ is as well off and agents $i_2$ and $i_3$ are strictly better off in $\mu^T$ than in $\mu^D$, we have that there is no horizon-$k$ improving path from $\mu^T$ to $\mu^D$ for $k \geq 5$. That is, $\mu^D \notin \phi_k(\mu^T)$ for $k \geq 5$. Hence, $\{\mu^D\}$ is not a horizon-$k$ vNM stable set for $k \geq 5$ since (ES) is violated. 

Computing the horizon-$k$ improving paths emanating from $\mu^T$ for $k \geq 5$, we get $\phi_k(\mu^T)=\{\mu^1,\mu^2,\mu^3,\mu^4\}$. Notice that $\mu^5 \notin \phi_k(\mu^T)$ since agent $i_1$ is worse off in $\mu^5$ than in $\mu^T$. From $\mu^1$, $\mu^2$, $\mu^3$, $\mu^4$ and $\mu^5$, there is a horizon-$k$ improving path to $\mu^D$. That is, $ \mu^D \in \phi_k(\mu)$ for $\mu \in \{\mu^1,\mu^2,\mu^3,\mu^4,\mu^5\}$. From $\mu^D$ there is only a horizon-$k$ improving path to $\mu^T$; i.e. $\phi_k(\mu^D)=\{\mu^T\}$ for $k \geq 5$.

Thus, we have that, for $k \geq 5$, (i) $\phi_k(\mu^D)=\{\mu^T\}$, (ii) $\mu^D \notin \phi_k(\mu^T)$, and (iii) $\mu^D \in \phi_k(\mu)$ for each $\mu$ such that $\mu \in \phi_k(\mu^T)$. It follows then that for $k \geq 5$, $V = \{\mu^T\}= \{\{(i_1,s_3),(i_2,s_1),(i_3,s_2)\}\}$ is the unique horizon-$k$ vNM stable set. First, any $V^{\prime} \neq \{\mu^T\}$ such that $\mu^T \in V$ violates IS. Hence, $\mu^T \notin V^{\prime}$. Second, any $V^{\prime} \nsupseteq \{\mu^T,\mu^D\}$ violates ES. Hence, $\mu^D \in V^{\prime}$. Third, any $V^{\prime} \neq \{\mu^D\}$ such that $\mu^D \in V^{\prime}, \mu^T \notin V^{\prime}$ violates either IS or ES.

Thus, the DA matching $\mu^D$ and the IA matching $\mu^B$ do not belong to any horizon-$k$ vNM stable set for $k \geq 5$. Since the matching obtained from the IA algorithm is Pareto efficient, Example \ref{Ex1} also shows that there are priority-based matching problems where some Pareto efficient matching does not belong to any horizon-$k$ vNM stable set for $k\geq 3 \gamma -1$.

\begin{aremark}
There are priority-based matching problems such that, for $k\geq 3 \gamma -1$,
\begin{itemize}
\item[(i)] the matching obtained from the Deferred Acceptance (DA) algorithm does not belong to any horizon-$k$ vNM stable set;
\item[(ii)] some Pareto efficient matching does not belong to any horizon-$k$ vNM stable set.
\end{itemize} 
\end{aremark}

What happens if $k$ becomes small? Computing the horizon-$k$ improving paths emanating from $\mu^T$ for $k \leq 4$ in Example \ref{Ex1}, we get $\phi_k(\mu^T)=\{\mu^1,\mu^2,\mu^3,\mu^4,\mu^5,\mu^D\}$. So, there is now a horizon-$k$ improving path from $\mu^T$ to $\mu^D$. In addition, from $\mu^1$, $\mu^2$, $\mu^3$, $\mu^4$ and $\mu^5$, there is still a horizon-$k$ improving path to $\mu^D$. That is, $ \mu^D \in \phi_k(\mu)$ for $\mu \in \{\mu^1,\mu^2,\mu^3,\mu^4,\mu^5,\mu^T\}$ for $k \leq 4$. From $\mu^D$ there is no horizon-$k$ improving path for $k \leq 2$, but there is one for $3 \leq k \leq 4$; i.e. $\phi_k(\mu^D)=\emptyset$ for $k \leq 2$ and $\phi_k(\mu^D)=\{\mu^T\}$ for $3 \leq k \leq 4$. It follows then that for $3 \leq k \leq 4$, both $V=\{\mu^D\}$ and $V^{\prime}=\{\mu^T\}$ are horizon-$k$ vNM stable sets. However, for $k \leq 2$, $V = \{\mu^D\}$ is the unique horizon-$k$ vNM stable set. In general, it holds that for any priority-based matching problem, the DA matching $\mu^D$ belongs to all horizon-$1$ vNM stable sets.\footnote{Proposition 3 in Luo, Mauleon and Vannetelbosch (2021) provides a characterization of a horizon-$1$ vNM stable set of networks. Since matchings are a subclass of networks and the DA matching is stable, it follows that the DA matching belongs to all horizon-$1$ vNM stable sets.}

\begin{aremark}
Let $\langle I,S,P,F\rangle$ be a priority-based matching problem and $\mu^D$ be the matching obtained from the Deferred Acceptance (DA) mechanism. The matching $\mu^D$ belongs to all horizon-$1$ vNM stable sets.
\end{aremark}

\section{Discussion}\label{S6}

\subsection{Tighter bound on $k$}

We now look whether one could find a tighter bound on $k$ such that for all $k^{\prime} \geq k$, the singleton set $\{\mu^T\}$ is a horizon-$k^{\prime}$ vNM stable set. Consider the proof of Theorem \ref{T5}. At the end of Step 1.1.A we reach the matching $\mu_{1,1}^{\prime } = \mu + \{(i,s) \mid i,s \in c_1^1 \text{ and } s \mapsto i\} - \{(i,\mu(i)) \mid i,s \in c_1^1 \text{, } s \mapsto i \text{ and } \mu(i) \neq s\} - \{(\mu^{-1}(s),s) \mid i,s \in c_1^1 \text{, } s \mapsto i \text{ and } \mu^{-1}(s) \neq i\}$ where each agent involved in $c_1^1$ is matched to the object that ranks her first on its priority list. Remember that Step 1.1.A counts at most $\bar{c}_1^1$ steps where $\bar{c}_1^1=\#\{i \in I \mid i \in c_1^1\}$ is the number of agents involved in cycle $c_1^1$. We slightly modify Step 1.1.B as follows.

\begin{itemize}
\item[Step 1.1.B]Let $\bar{\mathcal{I}}_1^1=\{(i^l)\}_{l=1}^{\bar{c}_1^1}$ be an ordered set of the agents involved in cycle $c_1^1$ such that $i^l \in c_1^1$ and $i^l$ matches before $i^{l+1}$ to the object that ranks her first on its priority list in Step 1.1.A, for $l=1,...,\bar{c}_1^1-1$. Agents who were already matched (at the beginning of Step 1.1.A ) to the objects that rank them first on their priority lists occupy the first positions of $\bar{\mathcal{I}}_1^1$. From the matching $\mu_{1,1}^{\prime }$, agents $i^1$ to $i^{\bar{c}_1^1-1}$ successively leave their objects to reach the matching $\mu_{1,1}^{\prime \prime} = \mu_{1,1}^{\prime} - \{(i,s) \mid i,s \in c_1^1 \text{, } s \mapsto i \text{ and } i \neq i^{\bar{c}_1^1}\} $ where only agent $i^{\bar{c}_1^1}$ is still matched to the object that ranks her first on its priority list. Step 1.1.B counts at most $\bar{c}^1_1 -1$ steps.

\item[Step 1.1.C.]From the matching $\mu_{1,1}^{\prime \prime}$, agent $i^{\bar{c}_1^1}$ first matches with her top choice object $s$ to reach the matching $\mu_{1,1}^{\prime \prime} + (i^{\bar{c}_1^1},s) - (i^{\bar{c}_1^1}, \mu_{1,1}^{\prime \prime}(i^{\bar{c}_1^1}))$ where $s$ is such that $i \mapsto s \text{ and } i,s \in c^1_1$. Notice that $s$ was unassigned at $\mu_{1,1}^{\prime \prime}$ while the object that ranks $i^{\bar{c}_1^1}$ first on its priority list, i.e. $\mu_{1,1}^{\prime \prime}(i^{\bar{c}_1^1})$, is now unassigned. Next, agents $i^1$ to $i^{\bar{c}_1^1-1}$ successively match to their top choice object to reach the matching $\mu_{1,1}^{\prime \prime \prime} = \mu_{1,1}^{\prime \prime} - \{(i,s) \mid i,s \in c_1^1 \text{ and } s \mapsto i \} + \{(i,s) \mid i,s \in c_1^1 \text{ and } i \mapsto s\}$. Step 1.1.C counts at most $\bar{c}^1_1$ steps. We have reached $\mu_{1,1}^{\prime \prime \prime}$ with $m_1^1 \subseteq \mu_{1,1}^{\prime \prime \prime}$ and so agents belonging to $c_1^1$ are assigned to the same object as in $\mu^T$. If $1 \neq L_1$, then go to Step 1.2. Otherwise, go to Step 1.End with $\mu_{1,L_1}^{\prime \prime \prime} = \mu_{1,1}^{\prime \prime \prime}$.

\end{itemize}

Step 1.1 counts at most $3 \bar{c}^1_1 - 1$ steps. Since we use now the order in which the agents are matched to the objects that rank them first on their priority lists in Step 1.1.A, there is at most $2 \bar{c}^1_1 + 1$ steps between the first move of an agent in Step 1.1.A and her final move in Step 1.1.C. Hence, it becomes sufficient that the agents who move in Step 1.1 look forward $2 \bar{c}^1_1 + 1$ steps ahead to have incentives for engaging the move towards the matching $\mu_{1,1}^{\prime \prime \prime}$ where they already get the object assigned by the TTC. Once they reach those matches they do not move afterwards.

Thus, the lower bound $\underline{k} = 2 \gamma + 1$ is a tighter bound on $k$ such that for all $k^{\prime} \geq \underline{k}$, the singleton set $\{\mu^T\}$ is a horizon-$k^{\prime}$ vNM stable set. However, it relies on improving paths that require much more coordination on behalf of the agents than the ones associated to the lower bound of Theorem \ref{T5}, i.e. $3 \gamma - 1$. 

\subsection{Farsighted stable set}

By simply replacing $\mu_{\min \{l+k,L\}}(i) P_i \mu _{l}(i)$ by $\mu_{L}(i) P_i \mu _{l}(i)$ in the definition of a horizon-$k$ improving path, we obtain the definition of a farsighted improving path. Let $\phi_{\infty}(\mu)$ be the set of matchings that can be reached by means of a farsighted improving path emanating from $\mu$. Given the number of possible matchings is finite, there exists $k^{\star}$ such that for all $k\geq k^{\star}$, $\phi_k(\mu)=\phi_{k+1}(\mu)$, and so $\phi_{k^{\star}}(\mu)=\phi_{\infty}(\mu)$.\footnote{Mauleon, Vannetelbosch and Vergote (2011) define and characterize the vNM farsighted stable set for two-sided matching problems.} 

\begin{adefinition} \label{D7} Let $\langle I,S,P,F\rangle$ be a priority-based matching problem. A set of matchings $V \subseteq \mathcal{M}$ is a vNM farsighted stable set if it satisfies:
\begin{itemize}
\item[(\textbf{i})] Internal stability (IS): For every $\mu, \mu^{\prime} \in V$,
it holds that $\mu^{\prime} \notin \phi_{\infty}(\mu).$
\item[(\textbf{ii})] External stability (ES): For every $\mu \in \mathcal{M} \setminus V$, it holds that $\phi_{\infty}(\mu) \cap V \neq \emptyset. $
\end{itemize}
\end{adefinition}

\begin{acorollary}
Let $\langle I,S,P,F\rangle$ be a priority-based matching problem and $\mu^T$ be the matching obtained from the Top Trading Cycles mechanism. The singleton set $\{\mu^T\}$ is a vNM farsighted stable set.\footnote{This result is robust to the incorporation of various forms of maximality in the definition of farsighted improving path, like the strong rational expectations farsighted stable set in Dutta and Vohra (2017) and absolute maximality as in Ray and Vohra (2019). See also Herings, Mauleon and Vannetelbosch (2020).} 
\end{acorollary}

\subsection{Horizon-$L$ farsighted set}

An alternative concept for limited farsightedness is obtained by adapting Herings, Mauleon and Vannetelbosch (2019) definition of a horizon-$L$ farsighted set of networks to priority-based matching problems. A horizon-$L$ farsighted set of matchings has to satisfy three conditions: (i) external deviations should be horizon-$L$ deterred, (ii) from any matching outside the set there is a sequence of farsighted improving paths of length smaller than or equal to $L$ leading to some matching in the set, (iii) there is no proper subset satisfying the first two conditions.

\begin{adefinition}\label{DFIPL} 
Let $\langle I,S,P,F\rangle$ be a priority-based matching problem. A farsighted improving path of length $L$ from a matching $\mu \in \mathcal{M}$ to a matching $\mu^{\prime } \in \mathcal{M} \setminus \{\mu \}$ is a finite sequence of distinct matchings $\mu_{0},\ldots ,\mu_{L}$ with $\mu_{0} = \mu$ and $\mu_{L} = \mu^{\prime }$ such that for every $l \in \{0,\ldots ,L-1\}$ either
\begin{itemize}
\item[(\textbf{i})]$ \mu_{l + 1} = \mu_{l} - (i,s)$ for some $(i,s) \in I \times S$ such that $\mu_{L}(i) P_i \mu _{l}(i)$, or
\item[(\textbf{ii})]$ \mu_{l + 1} = \mu_{l} + (i,s) - \{(i,\mu_{l}(i)) \mid \mu_{l}(i) \neq i\} - \{(\mu_{l}^{-1}(s),s) \mid \#(\mu_{l}^{-1}(s))=1\}$ for some $(i,s) \in I \times S$ such that $\mu_{L}(i) P_i \mu _{l}(i)$ and $F_s(i) < F_s(\mu_{l}^{-1}(s))$ if $\#(\mu_{l}^{-1}(s)) = 1$.
\end{itemize}
\end{adefinition}

If there exists a farsighted improving path of length $L$ from $\mu$ to $\mu^{\prime}$, then we write $\mu \longrightarrow_{L} \mu^{\prime}$. For a given matching $\mu$ and some $L^{\prime }\geq 1$, let $\widehat{\phi}_{L^{\prime }}(\mu)$ be the set of matchings that can be reached from $\mu$ by a farsighted improving path of length $L\leq L^{\prime }$. That is, $\widehat{\phi}_{L^{\prime }}(\mu)=\{\mu^{\prime }\in \mathcal{M \mid \exists }L \leq L^{\prime }$ such that $\mu \longrightarrow _{L}\mu^{\prime }\}$. Let $\widehat{\phi}_{\infty }(\mu)=\{\mu^{\prime
}\in \mathcal{M \mid \exists }L \in \mathbb{N}$ such that $\mu \longrightarrow _{L} \mu^{\prime }\} = \phi_{\infty}(\mu)$ denote the set of matchings that can be reached from $\mu$ by some farsighted improving path. From Lemma 1 in Herings, Mauleon and Vannetelbosch (2019) we have that for every $L \geq 1$, for every $\mu \in \mathcal{M}$, it holds that $\widehat{\phi}_{L}(\mu)\subseteq \widehat{\phi}_{L+1}(\mu)$, and that for $L\geq k^{\star}$, for every $\mu \in \mathcal{M}$, it holds that $\widehat{\phi}_{L}(\mu)=\widehat{\phi}_{L+1}(\mu)=\widehat{\phi}_{\infty }(\mu)=\phi_{\infty }(\mu)$.

The set $\widehat{\phi}_{L}^{2}(\mu)=\widehat{\phi}_{L}(\widehat{\phi}_{L}(\mu))=\{\mu^{\prime \prime }\in \mathcal{M}\mid \exists \mu^{\prime }\in \widehat{\phi}_{L}(\mu)$ such that $\mu^{\prime \prime }\in \widehat{\phi}_{L}(\mu^{\prime })\}$ consists of those matchings that can be reached by a composition of two farsighted improving paths of length at most $L$ from $\mu$. For $m \in \mathbb{N}$, let $\widehat{\phi}_{L}^{m}(\mu)$ be the matchings that can be reached from $\mu$ by means of $m$ compositions of farsighted improving paths of length at most $L$. Let $\widehat{\phi}_{L}^{\infty }$ denote the set of matchings that can be reached from $\mu$ by means of any number of compositions of farsighted improving paths of length at most $L$.\footnote{From Lemma 2 in Herings, Mauleon and Vannetelbosch (2019) we have that for every $L\geq 1$, for every $\mu \in \mathcal{M}$, it holds that $\widehat{\phi}_{L}^{\infty }(\mu)\subseteq \widehat{\phi}_{L+1}^{\infty }(\mu)$, and that for $L\geq k^{\star}$, for every $\mu\in \mathcal{M}$, it holds that $\widehat{\phi}_{L}^{\infty }(\mu)=\widehat{\phi}_{L+1}^{\infty }(\mu)=\widehat{\phi}_{\infty }^{\infty }(\mu)$.}

The notion of a horizon-$L$ farsighted set is based on two main requirements: horizon-$L$ deterrence of external deviations and horizon-$L$ external stability. A set of matchings $V$ satisfies horizon-$L$ deterrence of external deviations if all possible deviations from any matching $\mu \in V$ to a matching outside $V$ are deterred by a threat of ending worse off or equally well off.\footnote{We use the notational convention that $\widehat{\phi}_{-1}(\mu)=\emptyset $ for every $\mu \in \mathcal{M}$.}

\begin{adefinition}\label{Definition3a} 
For $L \geq 1$, a set of matchings $V \subseteq \mathcal{M}$ satisfies horizon-$L$ deterrence of external deviations if for every $\mu \in V$,

\begin{description}
\item[(i)] $\forall $ $(i,s) \notin \mu$ such that $\widetilde{\mu}=\mu + (i,s) - \{(i,\mu(i)) \mid \mu(i) \neq i\} - \{(\mu^{-1}(s),s) \mid \#(\mu^{-1}(s))=1\}$ and $\widetilde{\mu} \notin V$, either there exists $ \mu^{\prime } \in \lbrack \widehat{\phi}_{L-2}(\widetilde{\mu}) \cap V] \cup \lbrack \widehat{\phi}_{L-1} (\widetilde{\mu}) \setminus \widehat{\phi}_{L-2}(\widetilde{\mu})]$ such that $\mu R_i \mu^{\prime }$, or $F_s(i) > F_s(\mu^{-1}(s))$ if $\#(\mu^{-1}(s)) = 1$, 

\item[(ii)] $\forall $ $(i,s) \in \mu$ such that $\widetilde{\mu}=\mu - (i,s)$ and $\widetilde{\mu} \notin V$, there exists $ \mu^{\prime } \in \lbrack \widehat{\phi}_{L-2}(\widetilde{\mu}) \cap V] \cup \lbrack \widehat{\phi}_{L-1} (\widetilde{\mu}) \setminus \widehat{\phi}_{L-2}(\widetilde{\mu})]$ such that $\mu R_i \mu^{\prime }$.
\end{description}
\end{adefinition}

Condition (i) in Definition~\ref{Definition3a} captures that forming the match $(i,s)$ at $\mu\in V$ and reaching a matching $\widetilde{\mu}$ outside of $V$, is deterred by the threat of ending in $\mu^{\prime }$. Here $\mu^{\prime }$ is such that either there is a farsighted improving path of length smaller than or equal to $L-2$ from $\widetilde{\mu}$ to $\mu^{\prime }$ and $\mu^{\prime }$ belongs to $V$ or there is a farsighted improving path of length equal to $L-1$ from $\widetilde{\mu}$ to $\mu^{\prime }$ and there is no farsighted improving path from $\widetilde{\mu}$ to $\mu^{\prime }$ of smaller length.\footnote{We distinguish farsighted improving paths of length less than or equal to $L-2$ after a deviation from $\mu$ to $\widetilde{\mu}$ and farsighted improving paths of length equal to $L-1$. In the former case, the reasoning capacity of the agent is not yet reached, and the threat of ending in $\mu^{\prime }$ is only credible if it belongs to the set $V$. In the latter case, the only way to reach $\mu^{\prime }$ from $\mu$ requires $L$ steps of reasoning or even more (i.e. one step in the deviation to $\widetilde{\mu}$ and at least $L-1$ additional steps in any farsighted improving path from $\widetilde{\mu}$ to $\mu^{\prime }$). Since this exhausts the reasoning capacity of the agent, the threat of ending in $\mu^{\prime }$ is credible, irrespective of whether it belongs to $V$ or not.} 

A set of matchings $V$ satisfies horizon-$L$ external stability if from any matching outside of $V$ there is a sequence of farsighted improving paths of length smaller than or equal to $L$ leading to some matching in $V$.

\begin{adefinition}\label{Definition3b} For $L\geq 1$, a set of matchings $V\subseteq \mathcal{M}$
satisfies horizon-$L$ external stability if for every $\mu^{\prime }\in \mathcal{M}\setminus V$, $\widehat{\phi}_{L}^{\infty }(\mu^{\prime })\cap V\neq \emptyset $.
\end{adefinition}

\begin{adefinition}\label{Definition3} For $L\geq 1$, a set of matchings $V \subseteq \mathcal{M}$ is a horizon-$L$ farsighted set if it is a minimal set satisfying horizon-$L$ deterrence of external deviations and horizon-$L$ external stability.
\end{adefinition}

From Herings, Mauleon, and Vannetelbosch (2019) we have that a horizon-$L$ farsighted set of matchings exists.

\begin{atheorem}
Let $\langle I,S,P,F\rangle$ be a priority-based matching problem and $\mu^T$ be the matching obtained from the Top Trading Cycles mechanism. The singleton set $\{\mu^T\}$ is a horizon-$L$ farsighted set for $L\geq 3 \gamma $.
\end{atheorem}

\begin{proof}
Take $L\geq 3 \gamma $. We show that $\{\mu^T\}$ is a horizon-$L$ farsighted set. First, $\{\mu^T\}$ is a minimal set. Second, $\{\mu^T\}$ satisfies horizon-$L$ deterrence of external deviations. Any deviation from $\mu^T$ to $\widetilde{\mu}=\mu^T - (i,s)$ is deterred since agent $i$ is worse off at $\widetilde{\mu}$ where $\widetilde{\mu}(i)=i$. An agent $i$ may only have incentives to deviate from $\mu^T$ to $\widetilde{\mu}=\mu^T + (i,s) - \{(i,\mu(i)) \mid \mu(i) \neq i\} - \{(\mu^{-1}(s),s) \mid \#(\mu^{-1}(s))=1\}$ if she matches to an object $s$ that was assigned in $\mu^T$ to some agent $j$ who belongs to a cycle formed before agent $i$'s cycle in the TTC algorithm. Hence, any deviation from $\mu^T$ to $\widetilde{\mu}=\mu^T + (i,s) - \{(i,\mu(i)) \mid \mu(i) \neq i\} - \{(\mu^{-1}(s),s) \mid \#(\mu^{-1}(s))=1\}$ can be deterred since $\mu^T -(i,\mu^T(i)) \in \widehat{\phi}_{L-1}(\widetilde{\mu})$. Indeed, from $\widetilde{\mu}$ the agents belonging to agent $j$'s cycle in the TTC algorithm can simply follow the steps of Theorem \ref{T5}'s proof to reach $\mu^T -(i,\mu^T(i))$, and this farsighted improving path counts at most $3 \gamma -1$ moves.\footnote{For $L > 3 \gamma $, $\mu^T \in \widehat{\phi}_{L-1}(\widetilde{\mu})$ since agent $i$ has incentives to match to $s=\mu^T(i)$ at $\mu^T -(i,\mu^T(i))$.} Third, the horizon-$k$ improving path from Theorem \ref{T5}'s proof can be decomposed in a succession of farsighted improving paths of length smaller than or equal to $3 \gamma -1$ where each farsighted improving path consists of the formation of the matches between the agents belonging to the same cycle in the TTC algorithm. Hence, for every $\mu \in \mathcal{M}\setminus \{\mu^T\}$, $\mu^T \in \widehat{\phi}_{L}^{\infty }(\mu)$, and so $\{\mu^T\}$ satisfies horizon-$L$ external stability.
\end{proof}

\subsection{The agents own the objects}

Closely related to priority-based matching problems are matching problems where the agents own the object. Let $\langle I,S,P,F\rangle$ be a matching problem where each agent $i$ owns an object $s$. The strict priority structure $F$ of the objects over the agents is such that the priority $F_s$ of object $s$ only ranks the owner of object $s$. Without loss of generality, let agent $i_l$ be the owner of object $s_l$, for $l =1, \ldots ,n$. Let $i_s$ be a generic agent who owns object $s$.

\begin{aexample}\label{Ex2} 
Consider a matching problem $\langle I,S,P,F\rangle$ with $I=\{i_1,i_2,i_3\}$ and $S=\{s_1,s_2,s_3\}$, and where agent $i_l$ owns object $s_l$, for $l=1,2,3$. Agents' preferences and endowments are as follows.
\begin{center}
\begin{tabular}{cccc}
\multicolumn{4}{c}{Agents}\\
\hline
Endowment & $s_1$& $s_2$ & $s_3$  \\
\hline
& $P_{i_1}$& $P_{i_2}$ & $P_{i_3}$  \\
\hline
& $s_3$ & $s_3$ & $s_2$   \\
& $s_1$ & $s_2$ & $s_3$   \\
& $s_2$ & $s_1$ & $s_1$   
\end{tabular}
~\qquad~
\begin{tabular}{ccc}
\multicolumn{3}{c}{Objects}\\
\hline
$F_{s_1}$& $F_{s_2}$ & $F_{s_3}$\\
\hline
$i_1$ & $i_2$ & $i_3$\\
 &  &  \\
 &  & 
    \end{tabular}
  \end{center}

\end{aexample}

In Example \ref{Ex2}, $\mu^T=\{(i_1,s_1),(i_2,s_3),(i_3,s_2)\}$ is the matching obtained from the TTC algorithm. In the first round of the TTC algorithm, there is one cycle where agent $i_2$ points to object $s_3$, object $s_3$ points its owner $i_3$, agent $i_3$ points to object $s_2$ and object $s_2$ points its owner $i_2$. That is, $C_1=\{c_1^1\}$ with $c_1^1=(s_3,i_3,s_2,i_2)$. Agent $i_2$ is assigned to object $s_3$ and agent $i_3$ is assigned to object $s_2$: $m_1^1=\{(i_2,s_3),(i_3,s_2)\}$, and so $i_2$ and $i_3$ exchange their objects. In the second round of the TTC algorithm, there is only one leftover agent, $i_1$, who points to object $s_1$ that she owns and one leftover object, $s_1$, that points to its owner $i_1$. That is, $C_2=\{c_2^1\}$ with $c_2^1=(s_1,i_1)$. Agent $i_1$ is assigned to her own object $s_1$: $m_2^1=\{(i_1,s_1)\}$, and so $\mu^T=m_1^1 \cup m_2^1$. 

Roth and Postlewaite (1977) show that, for any matching problem $\langle I,S,P,F\rangle$ where each agent $i$ owns an object $s$, there is always a unique matching that is in the core. Moreover, this matching can be obtained with the TTC algorithm.\footnote{For matching problems with private endowments, Ma (1994) shows that a mechanism is strategy-proof, Pareto efficient and individually rational if and only if it uses the TTC algorithm.} 

\begin{adefinition}\label{DMFLIO} 
Let $\langle I,S,P,F\rangle$ be a matching problem where each agent $i$ owns an object $s$. A horizon-$k$ improving path from a matching $\mu \in \mathcal{M}$ to a matching $\mu^{\prime } \in \mathcal{M} \setminus \{\mu \}$ is a finite sequence of distinct matchings $\mu_{0},\ldots ,\mu_{L}$ with $\mu_{0} = \mu$ and $\mu_{L} = \mu^{\prime }$ such that for every $l \in \{0,\ldots ,L-1\}$ either
\begin{itemize}
\item[(\textbf{i})]$ \mu_{l + 1} = \mu_{l} - (i,s)$ for some $(i,s) \in I \times S$ such that $\mu_{\min \{l+k,L\}}(i) P_i \mu _{l}(i)$, or
\item[(\textbf{ii})]$ \mu_{l + 1} = \mu_{l} + (i,s) - \{(i,\mu_{l}(i)) \mid \mu_{l}(i) \neq i\} - \{(\mu_{l}^{-1}(s),s) \mid \#(\mu_{l}^{-1}(s))=1\}$ for some $(i,s) \in I \times S$ such that $\mu_{\min \{l+k,L\}}(i) P_i \mu _{l}(i)$ and $\mu_{\min \{l+k,L\}}(i_s) P_{i_s} \mu _{l}(i_s)$.
\end{itemize}
\end{adefinition}

In the case of matching problems where each agent owns an object, we still require that, along a horizon-$k$ improving path, each time some agent $i$ is on the move she is better off at the match she will get $k$-steps ahead on the sequence compared to her current match. Moreover, if agent $i$ matches to $s$, we also require that the owner of the object (i.e. $i_s$) prefers the match he will get $k$-steps ahead compared to his current match. In other words, the owner of the object has a word to say about the assignment of his endowment to some agent.
 
The set of matchings $\mu^{\prime }\in \mathcal{M}$ such that there is a horizon-$k$ improving path from $\mu$ to $\mu^{\prime }$ is denoted by $\widetilde{\phi}_k(\mu)$, so $\widetilde{\phi}_k(\mu)=\{\mu^{\prime} \in \mathcal{M} \mid \mu \rightarrow_k \mu^{\prime} \}$. Replacing $\phi_k(\mu)$ by $\widetilde{\phi}_k(\mu)$ in Definition \ref{D6} we obtain the definition of a horizon-$k$ vNM stable set for matching problems where each agent owns an object.

\begin{atheorem}
Let $\langle I,S,P,F\rangle$ be a matching problem where each agent $i$ owns an object $s$ and $\mu^T$ is the matching obtained from the Top Trading Cycles mechanism. The singleton set $\{\mu^T\}$ is the unique horizon-$k$ vNM stable set for $k \geq 3 \gamma -1$.
\end{atheorem}

\begin{proof}
From the proof of Theorem \ref{T5} it follows that $\{\mu^T\}$ is a horizon-$k$ vNM stable set for $k \geq 3 \gamma -1$.\footnote{Notice that along the horizon-$k$ improving path from the proof of Theorem \ref{T5}, an agent matches either to an unassigned object whose owner is unmatched or to an object that she is the owner. Hence, the owner of the object does not block her move towards the TTC matching $\mu^T$.}  We now show that $\{\mu^T\}$ is the unique horizon-$k$ vNM stable set for $k \geq 3 \gamma -1$ since $\widetilde{\phi}_k(\mu^T)=\emptyset$. Consider any cycle that is obtained in the first step of the TTC algorithm. All the agents involved in this cycle obtain their most preferred object in $\mu^T$ and, in this cycle, any agent obtains the endowment of another agent who is also in the cycle. Hence, from $\mu^T$ they will never engage a move towards another matching. Consider now any cycle that is obtained in the second step of the TTC algorithm. Taking as fixed the matches done in $\mu^T$ by all agents involved in any cycle of the first step of the TTC, all the agents involved in this cycle of the second step of the TTC obtain their most preferred object in $\mu^T$ and, in this cycle, any agent obtains the endowment of another agent who is also in the cycle. Knowing that agents from any cycle of the first step of the TTC will never engage a move, agents from any cycle of the second step of the TTC will never engage either a move from $\mu^T$ towards another matching. Repeating this argument with the matches found in steps $3,4, \ldots $ of the TTC leads to the conclusion that $\widetilde{\phi}_k(\mu^T)=\emptyset$. Hence, any set $V \neq \{\mu^T\}$ would violate (ES) or (IS), and so $\{\mu^T\}$ is the unique horizon-$k$ vNM stable set for $k \geq 3 \gamma -1$.
\end{proof}

\begin{acorollary}\label{COROOWNER} 
Let $\langle I,S,P,F\rangle$ be a matching problem where each agent $i$ owns an object $s$ and $\mu^T$ is the matching obtained from the Top Trading Cycles mechanism. The singleton set $\{\mu^T\}$ is the unique vNM farsighted stable set.\footnote{In an exchange economy with indivisible goods of Shapley and Scarf (1974), Kawasaki (2010) as well as Klaus, Klijn and Walzl (2010) show that there exists a unique vNM farsighted stable set, which coincides with the set of competitive allocations. Thus, they obtain a similar result to Corollary \ref{COROOWNER} except that they allow for coalitional moves while agents can only move one at a time in our definition of vNM farsighted stable set.} 
\end{acorollary}

\subsection{Conclusion}

We have considered priority-based matching problems with limited farsightedness. We have shown that, once agents are sufficiently farsighted, the matching obtained from the TTC algorithm becomes stable: a singleton set consisting of the TTC matching is a horizon-$k$ vNM stable set if the degree of farsightedness is greater than three times the number of agents in the largest cycle of the TTC. On the contrary, the matching obtained from the DA algorithm may not belong to any horizon-$k$ vNM stable set for $k$ large enough.

Hence, the TTC mechanism satisfies Pareto efficiency, strategy-proofness and (limited) farsighted stability. Notice that a mechanism is strategy-proof if no agent has incentives to misrepresent her preferences anticipating perfectly the outcome of the TTC algorithm. So, strategy-proofness implicitly presumes some degree of farsightedness on behalf of the agents. Thus, it seems more consistent to look for a mechanism that satisfies strategy-proofness together with (limited) farsighted  stability.

\section*{Acknowledgments}

Ana Mauleon and Vincent Vannetelbosch are, respectively, Research Director and Senior Research Associate of the National Fund for Scientific Research (FNRS). Financial support from the Fonds de la Recherche Scientifique - FNRS PDR research grant T.0143.18 is gratefully acknowledged. Ata Atay is a Serra H\'{u}nter Fellow (Professor Lector Serra H\'{u}nter). Ata Atay gratefully acknowledges financial support from the University of Barcelona through grant AS017672.



\section*{References}

\begin{description}


\item Abdulkadiro\u{g}lu, A., Y.K. Che, P. A. Pathak, A.E. Roth, and O. Tercieux (2020),
\textquotedblleft Efficiency, justified envy, and incentives in priority-based matching\textquotedblright , \textit{American Economic Review Insights} 2, 425-42.

\item Abdulkadiro\u{g}lu, A., and T. Sönmez (2003), \textquotedblleft School choice: A mechanism design approach\textquotedblright , \textit{American Economic Review} 93, 729-747.



\item Atay, A., A. Mauleon, and V. Vannetelbosch (2022), \textquotedblleft School choice with farsighted students\textquotedblright , CORE/LIDAM Discussion Paper 2022/25, UCLouvain, Louvain-la-Neuve, Belgium.




\item Che, Y.K., and O. Tercieux (2019), \textquotedblleft Efficiency and stability in large matching markets\textquotedblright , \textit{Journal of Political Economy} 127, 2301-2342.

\item Chwe, M. S.-Y. (1994), \textquotedblleft Farsighted coalitional stability\textquotedblright , \textit{Journal of Economic Theory} 63, 299-325.



\item Do\u{g}an, B., and L. Ehlers (2021), \textquotedblleft Minimally unstable Pareto improvements over deferred acceptance\textquotedblright , \textit{Theoretical Economics} 16, 1249-1279.

\item Do\u{g}an, B., and L. Ehlers (2022), \textquotedblleft Robust minimal instability of the top trading cycles mechanism\textquotedblright , \textit{American Economic Journal: Microeconomics} 14, 556-582.

\item Dubins, L.E., and D.A. Freedman (1981), \textquotedblleft Machiavelli and the Gale-Shapley algorithm\textquotedblright , \textit{American Mathematical Monthly} 88, 485-494.


\item Dutta, B., and R. Vohra (2017), \textquotedblleft Rational expectations and farsighted stability\textquotedblright , \textit{Theoretical Economics} 12, 1191-1227.

\item Ehlers, L. (2007), \textquotedblleft Von Neumann-Morgenstern stable sets in matching problems\textquotedblright ,\ \textit{Journal of Economic Theory} 134, 537-547.


\item Gale, D., and L.S. Shapley (1962), \textquotedblleft College admissions and the stability of marriage\textquotedblright , \textit{American Mathematical Monthly} 69, 9-15.


\item Haeringer, G. (2017), Market design: auctions and matching, MIT Press, Cambridge, MA.

\item Hakimov, R., and O. Kesten (2018), \textquotedblleft The equitable top trading cycles mechanism for school choice\textquotedblright , \textit{International Economic Review} 59, 2219-2258.



\item Herings, P.J.J., A. Mauleon, and V. Vannetelbosch (2017), \textquotedblleft Stable sets in matching problems with coalitional sovereignty and path dominance\textquotedblright , \textit{Journal of Mathematical Economics} 71, 14-19.

\item Herings, P.J.J., A. Mauleon, and V. Vannetelbosch (2019), \textquotedblleft Stability of networks under horizon-$K$ farsightedness\textquotedblright , \textit{Economic Theory} 68, 177-201.

\item Herings, P.J.J., A. Mauleon, and V. Vannetelbosch (2020), \textquotedblleft Matching with myopic and farsighted players\textquotedblright , \textit{Journal of Economic Theory} 190, 105125.

\item Kawasaki, R. (2010), \textquotedblleft Farsighted stability of the competitive allocations in an exchange economy with indivisible goods\textquotedblright , \textit{Mathematical Social Sciences} 9 (2010) 46-52.

\item Kesten, O (2006), \textquotedblleft On two competing mechanisms for priority-based allocation problems \textquotedblright , \textit{Journal of Economic Theory} 127, 155-171.

\item Kesten, O (2010), \textquotedblleft School choice with consent \textquotedblright , \textit{The Quarterly Journal of Economics} 125, 1297-1348.


\item Klaus, B., F. Klijn, and M. Walzl (2010), \textquotedblleft Farsighted house allocation\textquotedblright , \textit{Journal of Mathematical Economics} 46, 817-824.

\item Luo, C., A. Mauleon, and V. Vannetelbosch (2021), \textquotedblleft Network formation with myopic and farsighted players\textquotedblright , \textit{Economic Theory} 71, 1283-1317.

\item Ma, J. (1994), \textquotedblleft Strategy-proofness and the strict core in a market with indivisibilities\textquotedblright , \textit{International Journal of Game Theory} 23, 75-83. 


\item Mauleon, A., V. Vannetelbosch and W. Vergote (2011), \textquotedblleft Von Neumann - Morgenstern farsightedly stable sets in two-sided matching\textquotedblright , \textit{Theoretical Economics} 6, 499-521.

\item Morrill, T. (2015), \textquotedblleft Two simple variations of top trading cycles\textquotedblright , \textit{Economic Theory} 60, 123-140.





\item Ray, D. and R. Vohra (2015), \textquotedblleft The farsighted stable set\textquotedblright , \textit{Econometrica} 83, 977-1011.

\item Ray, D., and R. Vohra (2019), \textquotedblleft Maximality in the farsighted stable set\textquotedblright , \textit{Econometrica} 87, 1763-1779.

\item Reny, P.J. (2022), \textquotedblleft Efficient matching in the school choice problem\textquotedblright , \textit{American Economic Review} 112, 2025-43.

\item Roth, A.E. (1982), \textquotedblleft The economics of matching: stability and incentives\textquotedblright , \textit{Mathematics of Operations Research} 7, 617-628.

\item Roth, A.E. and A. Postlewaite (1977), \textquotedblleft Weak versus strong domination in a market with indivisible goods \textquotedblright , \textit{Journal of Mathematical Economics} 4, 131-137.

\item Roth, A.E. and M.A.O. Sotomayor (1990), Two-sided matching: a study in game-theoretic modeling and analysis, Econometric Society Monographs No.18, Cambridge University Press, Cambridge, UK.


\item Shapley, L.S., and H. Scarf (1974), \textquotedblleft On cores and indivisibility\textquotedblright , \textit{Journal of Mathematical Economics} 1, 23-37.


\end{description}

\end{document}